\documentclass{article}
\usepackage{amsmath, amsthm,amssymb,mathtools}
\usepackage{times}
\usepackage{enumitem}

\usepackage{graphicx}

\newcommand{\xcl}{LR-digraph\,\,}
\newcommand{\xcls}{LR-digraphs\,\,}
\newcommand{\xclssans}{LR-digraphs}
\newcommand{\xclsans}{LR-digraph}
\newcommand{\xclit}{{\em LR-digraph\,\,}}
\newcommand{\consec}{C1P-digraph\,\,}
\newcommand{\consecsans}{C1P-digraph}
\newcommand{\consecs}{C1P-digraphs\,\,}

\newcommand{\bmp}{\begin{minipage}}
\newcommand{\emp}{\end{minipage}}

\newcommand{\f}{\begin{it} mfvs\end{it}}
\newcommand{\fa}{\begin{it} amfvs\end{it}}
\newcommand{\tnae}{\begin{it} t_{NAE}\end{it}}
\newcommand{\ts}{\begin{it} t_{s}\end{it}}

\newcommand{\pathh}{\xrightarrow{*}}

\newtheorem{thm}{Theorem}
\newtheorem{cor}{Corollary}
\newtheorem{fait}{Claim}
\newtheorem{rmk}{Remark}
\newtheorem{defin}{Definition}
\newtheorem{ex}{Example}
\newtheorem{pb}{Problem}

\newcommand{\bp}{\begin{pb}\rm}
\newcommand{\ep}{\end{pb}}
\newcommand{\br}{\begin{rmk}\rm}
\newcommand{\er}{\end{rmk}}
\newcommand{\bdefin}{\begin{defin}\rm}
\newcommand{\edefin}{\end{defin} }
\newcommand{\bex}{\begin{ex}\rm}
\newcommand{\eex}{\end{ex}}

\newcommand{\bthm}{\begin{thm}}
\newcommand{\ethm}{\end{thm}}
\newcommand{\bcor}{\begin{cor}}
\newcommand{\ecor}{\end{cor}}
\newcommand{\bfn}{\begin{fait}}
\newcommand{\efn}{\end{fait}}

\usepackage{algorithm}
\usepackage{algorithmic}

\usepackage{xargs}                      

\begin{document}

\begin{center}
{\large\bf Min (A)cyclic Feedback Vertex Sets and Min Ones Monotone 3-SAT}
\bigskip

Irena Rusu \footnote{Email: Irena.Rusu@univ-nantes.fr}
\medskip

LS2N, UMR 6004, Universit\'e de Nantes\\ 2 rue de
la Houssini\`ere, BP 92208, 44322 Nantes, France
\end{center}
\bigskip\bigskip

\begin{abstract}
In directed graphs, we investigate the problems of finding: 1) a minimum feedback vertex set 
(also called the {\sc Feedback Vertex Set} problem,  or {\sc MFVS}),  2) a feedback vertex set inducing
an acyclic graph (also called the {\sc Vertex 2-Coloring without Monochromatic Cycles} problem, or {\sc Acyclic} FVS)
and 3) a minimum feedback vertex set inducing an acyclic graph ({\sc Acyclic MFVS}). 

We show that these  problems are strongly related to (variants of) {\sc Monotone 3-SAT}  and {\sc Monotone NAE 3-SAT}, 
where monotone means that all literals are in positive form. As a consequence, we deduce several NP-completeness
results on restricted versions of these problems. In particular, we define the {\sc 2-Choice} version of an 
optimization problem to be its restriction where the optimum value is known to be either $D$ or $D+1$ for some integer $D$, 
and the problem is reduced to decide which of $D$ or $D+1$ is the optimum value. We show that the {\sc 2-Choice} versions 
of  {\sc MFVS}, {\sc Acyclic MFVS}, {\sc Min Ones Monotone 3-SAT} and {\sc Min Ones Monotone NAE 3-SAT} 
are NP-complete. The two latter problems are the variants of {\sc Monotone 3-SAT}  and respectively {\sc Monotone NAE 3-SAT} 
requiring that the truth assignment minimize the number of variables set to true. 

Finally, we propose two classes of directed graphs for which {\sc Acyclic} FVS is polynomially solvable,
namely  flow reducible graphs (for which MFVS is already known to
be polynomially solvable) and  \consecs (defined by an adjacency matrix with the Consecutive Ones Property). 
\end{abstract}

\section{Introduction}

A {\em feedback vertex set} (abbreviated FVS) of a directed graph (or {\em digraph}) $G=(V,E)$ is a set $S\subsetneq V$ such that 
$S$ contains at least one vertex from each cycle of $G$.  Then we say that $S$ {\em covers} the cycles in $G$.
The term of {\em cycle cutset} or simply {\em cutset} 
is also used in the literature to name $S$. The {\sc Feedback Vertex Set} problem (abbreviated {\sc MFVS})
requires to find a FVS $S$ of minimum size in $G$.  
A {\em vertex 2-coloring without monochromatic cycle} of $G$ is a coloring of the vertices in $V$ 
with two colors, such that no cycle of $G$ has all its vertices of the same color. Each
of the colors thus defines a set of vertices inducing an acyclic graph, and each of them may therefore be seen as an 
acyclic FVS of $G$. The problem of deciding whether an acyclic FVS exists for a given digraph $G$ is classically called 
{\sc Vertex 2-coloring without monochromatic cycle}. We abbreviate it as {\sc Acyclic FVS}, in order to emphasize its
relationship with feedback vertex sets.

{\sc MFVS} has applications in path analysis of flowcharts of computer programs \cite{Shamir}, deadlock recovery in operating
systems \cite{Wang},
constraint satisfaction and Bayesian inference \cite{Bar}. The NP-completeness of {\sc MFVS} in directed graphs has been 
established in \cite{Karp}, and - given the reduction from {\sc Vertex Cover} - it implies both the APX-hardness of the problem,
and an extension of the results to undirected graphs.
The NP-completeness stands even for directed graphs with indegree and outdegree upper bounded by 2, as well as for 
planar graphs with indegree and outdegree upper bounded by 3 \cite{GJ}.
 Several classes of directed graphs for which MFVS is polynomially solvable
are known, including reducible flow graphs \cite{Shamir}, cyclically reducible graphs \cite{Wang},
quasi-reducible graphs \cite{Rosen} and completely contractible graphs \cite{Levy}. The best approximation
algorithms \cite{Seymour,Even} reach an approximation factor of  $O(\min\{log \tau^* log log \tau^*, logn\, log logn\})$, 
where $\tau^*$ is the optimal fractional solution in the natural LP relaxation of the problem. In 
the undirected case,  more intensively studied, many classes of graphs admitting polynomial solutions are known, 
as for instance interval graphs \cite{Lu}, permutation graphs~\cite{Brandstadt} and cocomparability graphs \cite{Coorg}.  
The best approximation algorithms \cite{Bafna,Becker} reach an approximation factor of 2. A review on MFVS
may be found in \cite{Festa}.

{\sc Acyclic FVS} has applications in micro-economics, and more particularly in the study of rationality
of consumption behavior \cite{deb2008efficient,Nobibon,Nobibon2}. The problem has been introduced in 
\cite{deb2008acyclic,deb2008efficient} together with a proof of NP-completeness in directed graphs.
But the problem is NP-complete even in simple directed graphs with no opposite arcs \cite{Nobibon} 
(called {\em oriented graphs}). A variant requiring that $S$ covers only the cycles
of a fixed length $k$ also turns out to be NP-complete for each $k\geq 3$ \cite{Karpinski}.
Very few particular classes of directed graphs for which the
problem becomes polynomial are known: line-graphs of directed graphs \cite{Nobibon2}, oriented planar graphs of maximum
degree 4 and oriented outerplanar graphs \cite{Nobibon}. In each graph of these classes, an acyclic FVS always exists. 
Several results not related to our work here
also exist on the undirected case but are in general devoted to the {\sc Vertex $k$-Coloring with No Monochromatic
Cycle problem} which is a generalization of {\sc Acyclic FVS} (see \cite{Nobibon} for more information).
By similarity with MFVS, we introduce here the problem {\sc Acyclic MFVS}, which requires to find a minimum size 
FVS that induces an acyclic graph, if such FVS exist.

In this paper, we consider only simple directed graphs (with no loops or multiple arcs). The NP-complete\-ness results
we prove concern simple directed graphs containing no pair of opposite arcs, {\em i.e.} oriented graphs. 
We define 3c-digraphs as a class in which it is sufficient to cover the 3-cycles 
in order to cover all the cycles. Starting with a set $\mathcal{C}$ of 3-literal clauses, we 
associate with $\mathcal{C}$ a set $\mathcal{F}$ of 3-literal clauses with literals in positive form,
and we associate with $\mathcal{F}$ a {\em representative graph} denoted $G^\triangleleft(\mathcal{F})$. 
Then we show various relationships between 3-SAT problems and FVS problems, according to the 
outline below. In this diagram, a {\em one} is a variable which is set to true.
\bigskip

{\small
\noindent\begin{tabular}{lp{0.25cm}lp{0.25cm}l}
 $\mathcal{C}$&$\rightarrow$&$\mathcal{F}$&$\rightarrow$&$G^\triangleleft(\mathcal{F})$\\ \hline
 set of 3-literal &&set of 3-literal clauses&& graph defined by the clauses, \\
 clauses&& with positive literals&&in $\mathcal{F}$, whose vertex set is $U$\\ 
 on variable set $X$& &on variable set $U$&& and which is a 3c-digraph\\ \hline
 satisfies 3-SAT &iff& satisfies {\sc Not-All-Equal 3-SAT}&&\\ \hline
 && satisfies 3-SAT with a truth&iff& $S$ is a FVS \\
 && assignment whose ones define the set $S\subsetneq U$ &&\\ \hline
 && satisfies {\sc Not-All-Equal 3-SAT} with a &iff& $S$ is an acyclic FVS\\
 && truth assignment whose ones define the set $S$&&\\ \hline
 
\end{tabular}}
\bigskip

Due to these results, MFVS and {\sc Acyclic MFVS} are close to the problems
{\sc Min Ones Monotone 3-SAT} and {\sc Min Ones Monotone Not-All-Equal 3-SAT}. We use these relationships to
deduce new hardness results on these problems. In particular, we define the {\sc 2-Choice} version of an 
optimization problem to be its restriction where the optimum value is known to be either $D$ or $D+1$ for some integer $D$, 
and the problem is reduced to decide which of $D$ or $D+1$ is the optimum value. We show that the {\sc 2-Choice} versions 
of  {\sc MFVS}, {\sc Acyclic MFVS}, {\sc Min Ones Monotone 3-SAT} and {\sc Min Ones Monotone NAE 3-SAT} are NP-complete.
In addition, we also show the NP-completeness of {\sc Acyclic FVS} even for 3c-digraphs.
The second part of the paper proposes two classes of graphs
for which {\sc Acyclic FVS} is polynomially solvable. More particularly, the graphs in these classes always
have an  acyclic FVS, and it may be found in polynomial time. These are the reducible flow graphs \cite{Hecht2} and the class
of  \consecs defined by an adjacency matrix with the Consecutive Ones Property.

The paper is organized as follows. In Section \ref{sect:notations} we give the main definitions and notations, and
precisely state the problems we are interested in. Section \ref{sect:reduction} presents the properties of the main reduction,
and gives a first NP-completeness result. Section \ref{sect:hardness} contains the other hardness results.  The
polynomial cases are studied in Section \ref{sect:classes}.
Section \ref{sect:conclusion} is the
Conclusion.

\section{Definitions and notations}\label{sect:notations}

We denote a directed graph as  $G=(V,E)$, and  its subgraph induced by a set $V'\subseteq V$ as $G[V']$.
Given an arc $vw$ from $v$ to $w$, $w$ is called a {\em successor} of $v$, and $v$ is called a  {\em predecessor} of $w$.
The set of successors (resp. predecessors) of $v$ is denoted $N^+(v)$ (resp. $N^-(v)$). Moreover,
$N^+[v]$ (resp. $N^-[v]$) is the notation for $N^+(v)\cup\{v\}$ (resp. $N^-(v)\cup\{v\}$).
A {\em 3-cycle digraph} (or a {\em 3c-digraph}) is a directed graph in which each cycle has three 
vertices defining a 3-cycle. (All the cycles used here are directed). 

Focusing now on instances of satisfiability problems, let $\mathcal{C}$ be a set of 3-literal clauses over a set of 
given variables $X=\{x_1, x_2, \ldots, x_n\}$. A literal is in {\em positive form} if it equals a variable, and
in {\em negative form}  otherwise. A truth assignment of the variables in $X$ such that each clause
has at least one true literal is called a {\em standard} truth assignment. When the truth assignment is such that
each clause has at least one true and at least one false literal, it is called a {\em not-all-equal (NAE)} truth assignment. 
Let $C_r=(l_a^r\vee l_b^r\vee l_c^r)$ be any clause of $\mathcal{C}$, and assume its literals 
are ordered from left to right. We define the (non-transitive) relation $\triangleleft$ as: 

\begin{equation}
l_a^r\triangleleft l_b^r, l_b^r \triangleleft l_c^r\, \hbox{and}\, l_c^r\triangleleft l_a^r
\label{eq:-1}
\end{equation}

\noindent where we assume that all clauses contribute to the same relation  $\triangleleft$, defined over all literals in $\mathcal{C}$.

Then we define $G^\triangleleft(\mathcal{C})$ to be the {\em representative graph} of $\mathcal{C}$ associated with the relation $\triangleleft$,
whose vertices are the literals and such that  $ll'\, \hbox{is an arc}$ iff the relation
$l\triangleleft l'$ has been established by at least one clause note that there are no multiple arcs).

A cycle of $G^\triangleleft(\mathcal{C})$ is called a {\em strongly 3-covered cycle} 
if it contains three vertices that are the three literals of some clause in $\mathcal{C}$.
These vertices thus induce a 3-cycle. 
Moreover, say that a set of 3-literal clauses $\mathcal{C}$, in which the order of literals is fixed, 
is in {\em strongly 3-covered form} if all the cycles of the graph $G^\triangleleft(\mathcal{C})$ associated 
with the relation $\triangleleft$  are strongly 3-covered cycles.  Note that not  all the sets $\mathcal{C}$ of 3-literal clauses
admit a strongly 3-covered form (examples are easy to build).

We present below the list of problems we are interested in. The notation $\ts(\mathcal{C})$ (resp. $\tnae(\mathcal{C})$) represents 
the  minimum number of true variables (or {\em ones}) in a standard (resp. NAE) truth assignment of  $\mathcal{C}$, 
if such an assignment exists; otherwise,   $\ts(\mathcal{C})$  (resp. $\tnae(\mathcal{C})$) equals the number of variables plus 1. Moreover, $\f(G)$ (resp. $\fa(G)$) is the cardinality of a 
minimum FVS (respectively a minimum acyclic FVS) of $G$. The {\em restrictions} mentioned in the list below
indicate the names of the subproblems we deal with, and which are explained below.
\bigskip

\noindent\begin{minipage}[t]{0.5\textwidth}
\noindent{\sc 3-SAT}\\
\noindent {\bf Input:} A set $\mathcal{C}$ of 3-literal clauses over a set of given variables.\\ 
\noindent {\bf Question:} Is there a standard truth assignment for the variables? 
 
\end{minipage}
\begin{minipage}[t]{0.5\textwidth}
 \hspace*{1cm} \hfill {{\bf Restrictions:}\\ \hspace*{1cm} \hfill M 3-SAT}\\
 \end{minipage}
 
\bigskip
\noindent\begin{minipage}[t]{0.5\textwidth}
\noindent{\sc Min Ones 3-SAT} ({\sc Min1 3-SAT}) \\
\noindent {\bf Input:} A set $\mathcal{C}$ of 3-literal clauses over a set of given variables. A variable $k$.\\ 
\noindent {\bf Question:} Is there a standard truth assignment for the variables with no more than $k$ true variables? 
\end{minipage}
\begin{minipage}[t]{0.5\textwidth}
\hspace*{1cm} \hfill{{\bf Restrictions:}\\ \hspace*{1cm} \hfill {\sc Min1-M 3-SAT}\\ \hspace*{1cm} \hfill {\sc 2-Choice-Min1-M 3-SAT}}
\end{minipage}
\bigskip

\noindent\begin{minipage}[t]{0.5\textwidth}
\noindent{\sc NAE 3-SAT} \\
\noindent {\bf Input:} A set $\mathcal{C}$ of 3-literal clauses over a set of given variables.\\ 
\noindent {\bf Question:} Is there a NAE truth assignment for the variables?
 \end{minipage}\begin{minipage}[t]{0.5\textwidth}
 \hspace*{1cm} \hfill {{\bf Restrictions:}\\ \hspace*{1cm} \hfill M-NAE 3-SAT}\\
 \end{minipage}
 
\bigskip
\noindent\begin{minipage}[t]{0.5\textwidth}
\noindent{\sc Min Ones NAE 3-SAT} ({\sc Min1-NAE 3-SAT})\\
\noindent {\bf Input:} A set $\mathcal{C}$ of 3-literal clauses over a set of given variables. A variable $k$.\\ 
\noindent {\bf Question:} Is there a NAE truth assignment for the variables with no more than $k$ true variables? 
 \end{minipage}\begin{minipage}[t]{0.5\textwidth}
 \hspace*{1cm} \hfill {{\bf Restrictions:}\\ \hspace*{1cm} \hfill {\sc Min1-M-NAE 3-SAT}\\  \hspace*{1cm} \hfill  {\sc 2-Choice-Min1-M-NAE 3-SAT}}\\
 \end{minipage}

 \bigskip
\noindent\begin{minipage}[t]{0.5\textwidth}
\noindent{\sc  Feedback Vertex Set } ({\sc MFVS }) \\
\noindent {\bf Input:} A directed graph $G=(V,E)$. A positive integer $k$.\\ 
\noindent {\bf Question:} Is it true that  $\f(G)\leq k$? 
 \end{minipage}\begin{minipage}[t]{0.5\textwidth}
 \hspace*{1cm} \hfill {{\bf Restrictions:}\\ \hspace*{1cm} \hfill {\sc 2-Choice-MFVS}}\\
 \end{minipage}
\bigskip

\noindent\begin{minipage}[t]{0.5\textwidth}
\noindent{\sc Acyclic FVS}\\
\noindent {\bf Input:} A directed graph $G=(V,E)$.\\ 
\noindent {\bf Question:} Is there a FVS $S\subseteq V$ such that\\ $G[V\setminus S]$ is acyclic? 
 \end{minipage}
 
\bigskip
\noindent\begin{minipage}[t]{0.5\textwidth}
\noindent{\sc Min Acyclic FVS} ({\sc Acyclic MFVS }) \\
\noindent {\bf Input:} A directed graph $G=(V,E)$. A positive integer $k$.\\ 
\noindent {\bf Question:} Is it true that $\fa(G)\leq k$? 
 \end{minipage}\begin{minipage}[t]{0.5\textwidth}
 \hspace*{1cm} \hfill {{\bf Restrictions:}\\ \hspace*{1cm} \hfill {\sc 2-Choice-Acyclic MFVS}}\\
 \end{minipage}
\bigskip

For each of the four satisfiability problems above we have the {\sc Monotone} restriction, asking that the clauses in the input be made 
only of  literals in positive form. These restrictions are identified by a supplementary M in the name of the problem (see the
first subproblem in the list of restrictions of each 3-SAT problem). The {\sc 2-Choice} restriction is defined for
each minimization problem. Here the input is reduced to instances for which the minimized parameter
is known to be either $D$ or $D+1$, for a given integer $D$, and the question is - as in the initial problem - whether the 
parameter is at most $D$ (or, equivalently, equal to $D$, given the hypothesis). The {\sc 2-Choice} restriction of the
{\sc Min1-M 3-SAT} problem is therefore:

\bigskip
\noindent{\sc 2-Choice Min1-M 3-SAT}\\
\noindent {\bf Input:} A set $\mathcal{C}$ of 3-literal clauses, all in positive form, over a set of given variables. An integer $D$
such that $D\leq \ts(\mathcal{C})\leq D+1$.\\
\noindent {\bf Question:} Is it true that $\ts(\mathcal{C})\leq D$? (equivalently, is it true that  $\ts(\mathcal{C})= D$?)
\bigskip

The {\sc 2-Choice} restrictions for the other minimization problems are similarly defined.

\begin{figure}
\centering
 \includegraphics[angle=0, width=11cm]{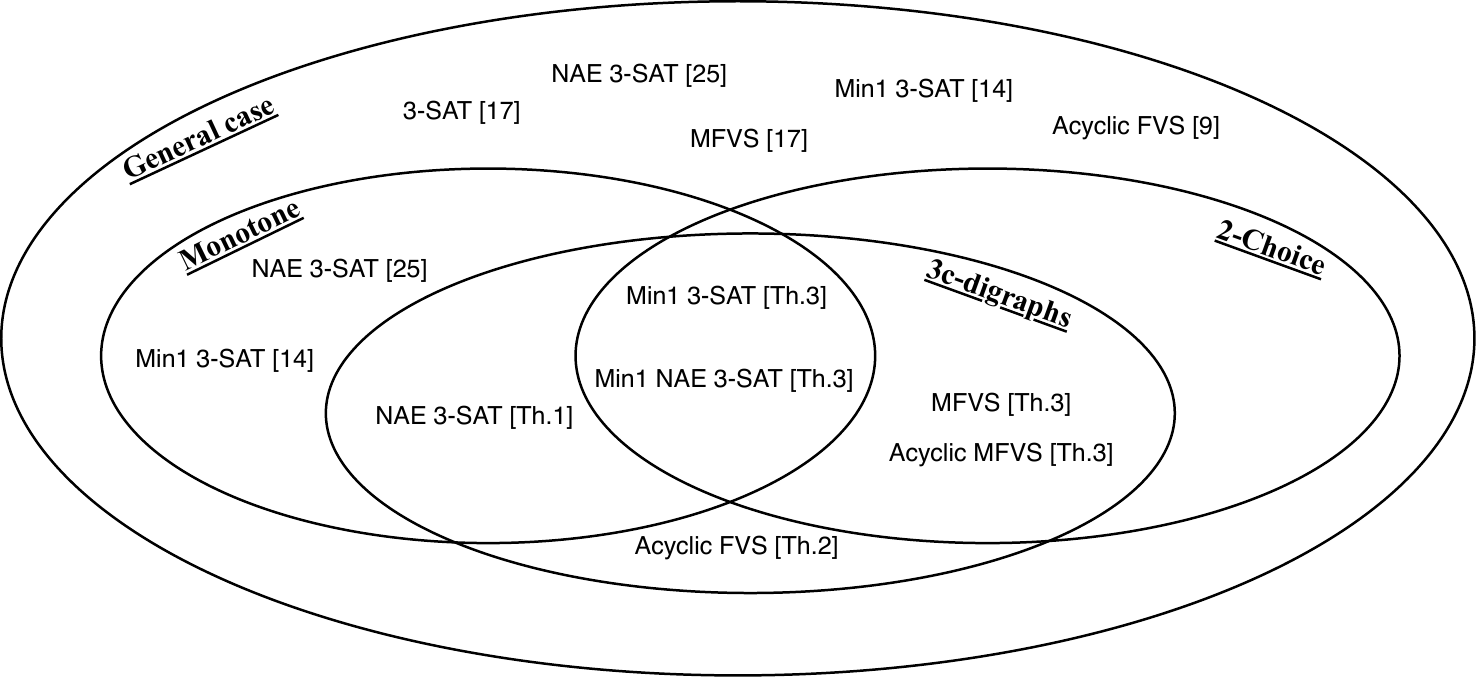}

 \caption{\label{fig:diagramme} Existing and new hardness results on the problems in our list. The external ellipse contains the problems in their most
 general form. Each internal ellipse represents a restriction: {\sc Monotone}, {\sc 2-Choice} or 3c-digraphs. Concerning
 the last restriction, it means that the input graph is a 3c-digraph for FVS problems, respectively that the representative graph 
 is a 3c-digraph for 3-SAT problems. Note that we always use the names of the initial (general) problems,
 the restrictions being deduced from inclusions into one or several ellipses.}
 \end{figure}

{\sc Fig.} \ref{fig:diagramme} indicates known NP-completeness results about the problems in our list, as well as the six 
NP-completeness results  we prove here (the six problems in the ellipse representing the restriction to 3c-digraphs). 
The NP-completeness of 3-SAT is due to Karp \cite{Karp}, whereas that of NAE 3-SAT and {\sc M-NAE 3-SAT} results from 
general theorems proved by Schaefer \cite{Schaefer}. The problem {\sc M 3-SAT} is obviously polynomial, but
{\sc Min1-M 3-SAT} is NP-complete (implying that Min1 3-SAT is NP-complete), since it is the same as 3-{\sc Hitting Set} and
{\sc Vertex Cover} for 3-Uniform Hypergraphs \cite{GJ}. As indicated above, MFVS and {\sc Acyclic} FVS are also NP-complete.

\section{From 3-SAT to M-NAE 3-SAT}\label{sect:reduction}

Our proofs are based on a classical sequence of reductions from 3-SAT to  {\sc M-NAE 3-SAT}, and on the good 
properties of the representative graph $G^\triangleleft(\mathcal{F})$ of the resulting set of clauses $\mathcal{F}$.

Let $\mathcal{C}$ be a set of $m$ 3-literal clauses, whose literals may be in positive or negative form, on the variable set 
$X=\{x_1, x_2,$ $\ldots, x_n\}$. We assume w.l.o.g. that no clause of 3-SAT 
contains the same variable twice (regardless to the positive or negative form). Otherwise, such clauses may be either removed (in
case both forms are present) or replaced with equivalent suitable clauses ($(l\vee l\vee l')$ for instance
may be replaced with $(l\vee l'\vee u)$ and $(l\vee l'\vee \overline{u})$
where $u$ is a new variable). Then in each clause $C_r=(l_i\vee l_j\vee l_k)$  we assume $i<j<k$ and
$l_u \in \{x_u, \overline{x_u}\}$, for all $u\in\{i, j, k\}$.  When  $\mathcal{C}$ is considered as an instance of  
3-SAT,  each clause $C_r=(l_i\vee l_j\vee l_k)$ is successively replaced with
sets of clauses, as in Table \ref{table:red}, so as to successively reduce 3-SAT to 
NAE 4-SAT, to NAE 3-SAT and to M-NAE 3-SAT. Note that NAE 4-SAT is stated similarly to NAE 3-SAT except that
the clauses have four literals.

\begin{table}[t]

\centering {\footnotesize
\begin{tabular}{|l|l|l|l|}
\hline
{\bf Reduction}&{\bf New variable set}&{\bf Set of clauses replacing clause}&{\bf Precisions} \\
&&\hspace*{1cm}$C_r=(l_i\vee l_j\vee l_k)$&\\ \hline
3-SAT to NAE 4-SAT& $U_1=\{y_1, y_2, \ldots, y_n,z\}$&$(h_i\vee h_j\vee h_k\vee z)$&$h_u=y_u$ if $l_u=x_u$\\
&&&$h_u=\overline{y_u}$ if $l_u=\overline{x_u}$\\ \hline
NAE 4-SAT to NAE 3-SAT& $U_2=\{y_1, y_2, \ldots, y_n,z,$ &$(h_i\vee h_j\vee w_r), (\overline{w_r}\vee h_k\vee z)$&\\ 
&\hspace*{1cm}$ w_1, w_2, \ldots, w_m\}$&&\\ \hline
NAE 3-SAT to& $U_3=\cup_{g\in U_2\setminus\{z\}}\{\alpha_g, \beta_g,  $&(Basic clauses)&$\gamma_u=\alpha_{y_u}$ if $h_u=y_u$\\
 M-NAE 3-SAT&$a_g, b_g, c_g\}\cup \{z\}$ &$(\gamma_i\vee \gamma_j\vee \alpha_{w_r}),(\beta_{w_r}\vee \gamma_k\vee z),$ &$\gamma_u=\beta_{y_u}$ if $h_u=\overline{y_u}$\\
 &&(Consistency clauses)&\\
& & $\cup_{g\in U_2\setminus\{z\}}\{(\alpha_g\vee \beta_g\vee a_g),$& One occurrence. \\
&& $ (\alpha_g\vee \beta_g\vee b_g), (\alpha_g\vee \beta_g\vee c_g),$& (do not duplicate \\
&& $ (a_g\vee b_g\vee c_g)\}$& for each clause)\\ \hline

\end{tabular}
}

\caption{\label{table:red} Successive reductions from 3-SAT to M-NAE 3-SAT. Explanations are given in the text.}
\end{table}

These reductions are explained as follows. In the first one (3-SAT to NAE 4-SAT), each variable $y_u$ is related with $x_u$   
in the sense that $x_u$ is assigned the value true in the 3-SAT instance iff the truth assignment in 
NAE 4-SAT is such that $y_u\neq z$. Equivalently, $l_u$ is true in the 3-SAT instance iff $h_u\neq z$ in the NAE 4-SAT instance. 
Furthermore, a new variable $w$ is added for each clause when the transition from
4-literal clauses  to 3-literal clauses is performed, making that each 4-literal clause is replaced with an equivalent 
pair of 3-literal clauses. Finally, in order to eliminate literals in negative form in the reduction from NAE 3-SAT to M-NAE 3-SAT,
each variable $g\in U_2\setminus\{z\}$ (since $z$ is already present only in positive form) is replaced with two variables $\alpha_g$ and $\beta_g$, 
respectively representing its literals in positive and negative form. Then,
each clause from NAE 3-SAT is replaced with its corresponding clause, a literal in positive (resp. negative) form being replaced with
its corresponding $\alpha_g$ (resp. $\beta_g$)  literal in positive form. 
The consistency clauses added for each variable guarantee that each pair of variables $\alpha_g$ and $\beta_g$ must have different values
in a NAE truth assignment. 

The final set $\mathcal{F}$ of 3-literal clauses (whose literals are all in positive form) associated to the initial
set $\mathcal{C}$ of 3-literal clauses is then:

\begin{multline}
\label{eq:1}
\mathcal{F}=\cup_{C_r=(l_i\vee l_j\vee l_k), C_r\in\mathcal{C}}\{\underbrace{(\gamma_i\vee \gamma_j\vee \alpha_{w_r})}_{F_r},\underbrace{(\beta_{w_r}\vee \gamma_k\vee z)}_{F'_r}\}\,\, \cup\\
\hspace*{2cm}\cup_{g\in U_2\setminus\{z\}}\{\underbrace{(\alpha_g\vee \beta_g\vee a_g)}_{F_{1g}}, \underbrace{(\alpha_g\vee \beta_g\vee b_g)}_{F_{2g}}, \underbrace{(\alpha_g\vee \beta_g\vee c_g)}_{F_{3g}}, \underbrace{(a_g\vee b_g\vee c_g)\}}_{F_{4g}}
\end{multline}
\medskip

\noindent where each $\gamma_u$ is either $\alpha_{y_u}$ or $\beta_{y_{u}}$ ($1\leq u\leq n$) according to the rules in
Table \ref{table:red}. In order to fix the terminology,
the four consistency clauses associated with a variable $g$ in $U_2\setminus\{z\}$ are denoted $F_{1g}$, $F_{2g}$, $F_{3g}$, $F_{4g}$ from left to right,
whereas the two basic clauses $(\gamma_i\vee \gamma_j\vee \alpha_{w_r})$ and $(\beta_{w_r}\vee \gamma_k\vee z)$ for some fixed $r\in\{1, 2, \ldots,m\}$ 
are respectively denoted $F_r$ and $F'_r$. We also denote by $N=|U_3|$ the total number of variables in $\mathcal{F}$,
which therefore satisfies $N=5(n+m)+1$.

The set $\mathcal{F}$ of clauses resulting from $\mathcal{C}$ in this way is said to be the {\em M-NAE version} of $\mathcal{C}$.
The successive equivalences between problems in Table \ref{table:red} (with their corresponding inputs)  are not very
difficult to show (see \cite{3sat} or \cite{moore} for a proof) and imply that:

\begin{fait}
 3-SAT with input $\mathcal{C}$ has a solution  iff M-NAE 3-SAT with input $\mathcal{F}$ has a solution.
\label{claim:SATred}
 \end{fait}

Or, equivalently, $\mathcal{C}$ has a standard truth assignment iff $\mathcal{F}$ has a NAE truth assignment.
Now, we fix  for each clause of $\mathcal{F}$ the order of literals used in Equation (\ref{eq:1}) (from left to right, in each clause). 
Define the relation $\triangleleft$ according to Equation(\ref{eq:-1}) and consider the representative graph $G^\triangleleft(\mathcal{F})$. 
Note that, since $\mathcal{F}$ has only literals in positive form, the set of vertices of $G^\triangleleft(\mathcal{F})$ 
is $U_3$.

\begin{fait}
Let $\mathcal{C}$ be any set of 3-literal clauses on a variable set, and let $\mathcal{F}$ be its 
M-NAE version. The representative graph  $G^\triangleleft(\mathcal{F})$ is oriented, and each of its cycles is strongly 3-covered. 
\label{claim:tiny}
\end{fait}

\begin{proof}  Recall that $i<j<k$ in each initial clause $C_r\in\mathcal{C}$.
We prove several affirmations, which lead to the conclusion.

\medskip

(A1) {\em $G^\triangleleft(\mathcal{F})$ contains no pair of opposite arcs and no loop.}
\medskip

As the variables involved in each clause of $\mathcal{C}$ (and thus of $\mathcal{F}$) are distinct, $G^\triangleleft(\mathcal{F})$ has no loop.
Furthermore - due to the local role of $a_g, b_g, c_g$ - the only arcs defined by the consistency clauses that could share both
endpoints with other arcs are $\alpha_g\beta_g$. But no basic clause contains literals with the same index $g$ (by the
same previous assumption), therefore
the arc $\alpha_g\beta_g$ cannot have an opposite arc. Focusing now exclusively on arcs defined by basic clauses,
note that basic clauses too contain literals with a local role - namely $\alpha_{w_r}, \beta_{w_r}$ - and therefore
only the arcs $\gamma_i\gamma_j$ and $\gamma_k z$ could possibly have opposite arcs. But for the former one this is impossible by the
assumption that $i<j$ in each clause, whereas the latter is impossible since the only arcs with source $z$ have destination
$\beta_{w_r}$ ($r=1, 2, \ldots, m$) and these literals are distinct from literals $\gamma_u$ ($u=1, 2, \ldots, n$).
\medskip

(A2) {\em Assume $Q=q_1q_2\ldots q_d$ is a shortest cycle of $G^\triangleleft(\mathcal{F})$ which is not
strongly 3-covered. Then we have 
\begin{itemize}
\item[(i)] $Q$ cannot contain any of the vertices $a_g, b_g$ or $c_g$ ($g\in U_2\setminus\{z\}$). 
\item[(ii)] the arcs of $Q$ possibly defined by consistency clauses are  $\alpha_g\beta_g$, $g\in U_2\setminus\{z\}$.
\item[(iii)] $Q$ cannot contain $z$.
\end{itemize}}

To show {\em (i)}, we notice that  the (arcless) subgraph of $G^\triangleleft(\mathcal{F})$ induced by  $a_g, b_g$ and $c_g$, 
for a fixed $g$, has ingoing arcs only from $\beta_g$ and outgoing arcs only to $\alpha_g$, implying that the cycle
should necessarily contain $\alpha_g$ and $\beta_g$, additionally to $a_g, b_g$ or $c_g$. Therefore the three literals of
a clause among $F_{1g}, F_{2g}$ and $F_{3g}$ would be contained in $Q$. This is in contradiction with the assumption that $Q$ is not strongly 3-covered.
Affirmation {\em (ii)} is an easy consequence of {\em (i)}. Finally, if Affirmation {\em (iii)} was false, then the successor
of $z$ would be some $\beta_{w_t}$  (arc defined by clause $F'_{t}$) since
the only successors of $z$ are of this form  and thus - taking into account that $Q$ cannot contain all literals of $F'_t$- we would have that the next literal
along the cycle is $a_{w_{t}}$ or $b_{w_{t}}$ or $c_{w_{t}}$ (from the consistency clauses $F_{1t}, F_{2t}$ or $F_{3t}$) 
thus contradicting {\em (i)}. 

\medskip

(A3) {\em If a  cycle of  $G^\triangleleft(\mathcal{F})$ contains at least two distinct literals
$\gamma_e\in \{\alpha_e,\beta_e\}$ and $\gamma_f\in\{\alpha_f,\beta_f\}$ with $e,f\in\{1, 2, \ldots, n\}$, then it is a
strongly 3-covered cycle.}
\medskip

Note that, by Affirmation (A1), such a cycle has at least three vertices. By contradiction, assume that $Q=q_1q_2\ldots q_d$ ($d\geq 3$)
is a shortest cycle of $G^\triangleleft(\mathcal{F})$ which is not
a strongly 3-covered cycle and contains $\gamma_e\in \{\alpha_e,\beta_e\}$ and $\gamma_f\in\{\alpha_f,\beta_f\}$ with $\gamma_e\neq \gamma_f$. In the case
where vertices $\gamma_e$ and $\gamma_f$ with $e\neq f$ exist on $Q$, choose $\gamma_e$ and $\gamma_f$
such that $e< f$ and the path $P$ from $\gamma_f$ to $\gamma_e$ along the cycle $Q$ is as short as possible. Then $P$
contains no other vertex $\gamma_h\in\{\alpha_h,\beta_h\}$, $h\in\{1, 2, \ldots, n\}$, since otherwise
$\gamma_e$ or $\gamma_f$ would have been differently chosen. Moreover, $P$ has at least two arcs, otherwise 
an arc $\gamma_f\gamma_e$ would exist with $f>e$ and no clause in $\mathcal{F}$ allows to build such an
arc, because of the convention that $i<j<k$ in each clause of $\mathcal{C}$. In the opposite case,
{\em i.e.} when only two vertices $\gamma_e$ and $\gamma_f$ with $e=f$ exist on $Q$, we necessarily
have that $\{\gamma_e,\gamma_f\}=\{\alpha_e,\beta_e\}$ and $\alpha_e\beta_e$ is an arc of $Q$, otherwise $Q$ is not as short as possible.  Then we
denote $P$ the path from $\beta_e$ (denoted $\gamma_f$ for homogenization reasons) to $\alpha_e$ (denoted $\gamma_e$).
We have again that $P$ contains no other vertex $\gamma_h\in\{\alpha_h,\beta_h\}$, and $P$ contains at least two
arcs (since $d\geq 3$).

Thus in all cases we have a subpath $P$ of $Q$ with at least two arcs, joining $\gamma_f$ to $\gamma_e$,
such that $e\leq f$ and there is no other $\gamma_h$, $h=1, 2, \ldots, n$, on $P$.
We may assume that $\gamma_f=q_p$ and $\gamma_e=q_s$, with $p<s$. Then, recalling that $P$ has at least two arcs,
we deduce that $q_{p+1}=\alpha_{w_r}$ (due to some basic clause $F_r=(\gamma_i\vee \gamma_f\vee \alpha_{w_r})$ with an appropriate $i$)  or $q_{p+1}=z$ (due
to some basic clause $F'_{s}=(\beta_{w_s}\vee \gamma_f\vee z)$). The latter case is impossible by Affirmation (A2.{\em iii}).
In the former case ($q_{p+1}=\alpha_{w_r}$), we deduce that $q_{p+2}=\beta_{w_r}$ (the only other successor of $\alpha_{w_r}$ is $\gamma_i$ from
the same clause $F_r$, but then $Q$ would be a strongly 3-covered cycle) and according to (A2.{\em ii}) $q_{p+3}$ cannot be defined by a consistency clause 
indexed $w_r$. Therefore
$q_{p+3}=\gamma_k$ from the clause $F'_r=(\beta_{w_r}\vee \gamma_k\vee z)$ and therefore $\gamma_k=\gamma_e$ (since there is no other literal of this form on 
the path from $\gamma_f$ to $\gamma_e$). As $F_r$ and $F'_r$ express the initial 3-SAT clause 
$(l_i\vee l_f\vee l_k)$, we deduce $i<f<k$. With $k=e$ we then have $f<e$ and this contradicts our choice above.
\medskip

(A4) {\em Every cycle in $G^\triangleleft(\mathcal{F})$ is a strongly 3-covered cycle.}
\medskip

By contradiction, assume that $Q=q_1q_2\ldots q_d$ ($d\geq 3$, by (A1)) is a shortest cycle of $G^\triangleleft(\mathcal{F})$ which is not
a strongly 3-covered cycle. By affirmation (A2.{\em iii}), $z$ does not belong to $Q$, and thus the basic clauses $F'_r$ can only
contribute to $Q$ with arcs of type $\beta_{w_r}\gamma_k$, for some $r$ and $k$. By affirmation (A3) which
guarantees that at most one vertex $\gamma_h$ belongs to $Q$, we deduce
that basic clauses $F_s$ can only contribute with arcs of type $\gamma_j\alpha_{w_s}$ or $\alpha_{w_s}\gamma_i$, for  
some $i, j, s$. We also have by (A3) that exactly 0 or 2 arcs, among all arcs of these three types - namely $\beta_{w_r}\gamma_k$,
$\gamma_j\alpha_{w_s}$ and $\alpha_{w_s}\gamma_i$, for all possible $r, i, j, k, s$ - exist in $Q$, since 0 or 1 occurrence
of a literal of type $\gamma_e$ is possible on $Q$. In the case 0 arc is accepted, then no arc from basic clauses is admitted on $Q$ and therefore
$Q$ cannot exist (since consistency clauses form only strongly 3-covered cycles). In the case two arcs are accepted, exactly one
vertex $\gamma_e$ exists in $Q$. One of the two arcs incident with it is 
$\gamma_e\alpha_{w_u}$ (for some clause $F_u=(\gamma_i\vee\gamma_e\vee \alpha_{w_u})$) since this is the only type of admitted arcs outgoing from $\gamma_e$. The arc ingoing
to $\gamma_e$  is either  $\beta_{w_v}\gamma_e$ (for some clause $F'_v=(\beta_{w_v}\vee \gamma_e\vee z)$) or $\alpha_{w_z}\gamma_e$ (for 
some clause $F_z=(\gamma_e\vee\gamma_j\vee \alpha_{w_z})$). All the other arcs are from consistency clauses.
The successor of $\alpha_{w_u}$ is then necessarily $\beta_{w_u}$ and Affirmation (A2.{\em ii}) implies we have to use
an arc outgoing from $\beta_{w_u}$ and defined by a basic clause. The only solution avoiding to use a third arc from basic clauses
is that the predecessor of $\gamma_e$ on $Q$ is $\beta_{w_v}$ and $v=u$ (so that $\beta_{w_u}=\beta_{w_v}$). But then $\gamma_e$
appears both in $F_u$ and in $F'_u$, and this is impossible since $F_u$ and $F'_u$ correspond to the same initial 3-SAT clause,
which has three distinct literals.
\end{proof}
\bigskip

As a consequence, $G^\triangleleft(\mathcal{F})$ is a 3c-digraph on the set of variables $U_3$, with strongly 3-covered cycles.
Let us give the full name {\sc Strongly 3-Covered Monotone NAE 3-SAT} to the problem {\sc M-NAE 3-SAT}
reduced to instances where the set of clauses admits a strongly 3-covered form. Then, the representative graph of the input
is a 3c-digraph. We have:

\begin{thm}
{\sc Strongly 3-Covered Monotone NAE 3-SAT } is NP-complete.
\end{thm}

\begin{proof} The problem is in NP since it is a particular case of {\sc M-NAE 3-SAT}. The reduction
is from 3-SAT, by associating to a set $\mathcal{C}$ of 3-literal clauses its M-NAE version $\mathcal{F}$. 
The set $\mathcal{F}$ is in strongly 3-covered form, by Claim~\ref{claim:tiny}. By Claim~\ref{claim:SATred},
the theorem follows. \end{proof}
\bigskip

\section{Hardness results}\label{sect:hardness}

This section first presents results showing the relationships between solutions of the various 3-SAT problems
and the FVS problems (Subsections \ref{subsect:standard} and \ref{subsect:nae}). Then these results are used
to deduce the aforementioned NP-completeness results (Subsection \ref{subsect:NP}).

\subsection{Properties of standard truth assignments}\label{subsect:standard}

\begin{fait}
 Let $S$ be a set of variables from $U_3$. Then $S$ is the set of true variables in a standard truth assignment 
 of $\mathcal{F}$ iff $S$ is a FVS of $G^\triangleleft(\mathcal{F})$ or $S=U_3$.
 \label{claim:standard}
\end{fait}

\begin{proof} By Claim~\ref{claim:tiny}, the set $S$ of true variables in a standard truth assignment for $\mathcal{F}$ allows us to
cover all the cycles in $G^\triangleleft(\mathcal{F})$. If $S\neq U_3$, we have that $S$ is a FVS.
Conversely, if $S$ is a FVS of $G$ or $S=U_3$, $S$ covers all the 3-cycles and thus  contains at least one
literal in each clause. Thus the truth assignment that sets to true all the variables in $S$ is a standard
truth assignment for $\mathcal{F}$. \end{proof}

\begin{rmk}
 Since $\mathcal{F}$ has only literals in positive form, a standard truth assignment for $\mathcal{F}$ always exists.
 \label{rmk:pos}
\end{rmk}

\begin{fait}
Consider a standard truth assignment of $\mathcal{F}$ with a minimum number of true variables. Then:

\begin{enumerate}
 \item[i)] for each $g\in U_2\setminus\{z\}$ at least one of $\alpha_g$ and $\beta_g$ is true.
 \item[ii)] for each $g\in U_2\setminus\{z\}$, exactly one of the variables $a_g, b_g$ and $c_g$ is true. 
\end{enumerate}

\noindent Therefore $\frac{2(N-1)}{5}\leq \ts(\mathcal{F})\leq \frac{2(N-1)}{5}+1$.
\label{claim:mintruth}
\end{fait}

\begin{proof} Let $ta()$ be a standard truth assignment of $\mathcal{F}$ with a minimum number of true variables.

Then no clause of type $(a_g\vee b_g\vee c_g)$, for some $g$, contains more than one true literal. 
If, by contradiction, such a clause has all its variable set to true, then we can define a new truth assignment
 $tb()$ as a variant of $ta()$ where variables $b_g$ and $c_g$ are set to false and $\alpha_g$ is set to true (if
 it is not already true). Then
 $tb()$ has less true variables than $ta()$, a contradiction. And in the case where there are exactly two true literals
 in the clause $(a_g\vee b_g\vee c_g)$, the false literal implies that at least one of the variables $\alpha_g, \beta_g$ needs to be true so as
to satisfy all the clauses $F_{1g}, F_{2g}$ and $F_{3g}$. But then one of the two true literals in $(a_g\vee b_g\vee c_g)$ 
may be set to false, implying that $ta()$ does not have a minimum number of true literals, a contradiction.

Then $ta()$ sets exactly one of the variables  $a_g, b_g, c_g$ to true, resulting into $\frac{N-1}{5}$ true variables. Moreover,
at least one of the variables $\alpha_g, \beta_g$ needs to be true so as to satisfy all the clauses $F_{1g}, F_{2g}$ and $F_{3g}$. 
This yields $\frac{N-1}{5}$ supplementary true variables, and fixes the lower bound of $\frac{2(N-1)}{5}$.
In order to show the upper bound,  consider the truth assignment that defines  $z$ and, for all 
$g\in U_2\setminus\{z\}$, variables $\alpha_g, a_g$ as true. All the other
variables are false. This is a standard truth assignment for $\mathcal{F}$ with $\frac{2(N-1)}{5}+1$ true variables. \end{proof}
\bigskip

\begin{fait}
  $\frac{2(N-1)}{5}\leq \f(G^\triangleleft(\mathcal{F}))=\ts(\mathcal{F})\leq \frac{2(N-1)}{5}+1$.
  \label{claim:ts=f}
\end{fait}

\begin{proof} By Claim~\ref{claim:standard},  a minimum FVS $S$ corresponds to the minimum set of true variables in a  standard truth 
assignment of $\mathcal{F}$, and viceversa. Therefore 
$\f(G^\triangleleft(\mathcal{F}))=\ts(\mathcal{F})$.
By Claim~\ref{claim:mintruth},  $\ts(\mathcal{F})$ has the required bounds. \end{proof}

\subsection{Properties of NAE truth assignments}\label{subsect:nae}

Let $G=(V,E)$ be an oriented graph, and $S$ an acyclic FVS of it. We start this section with a result which allows us 
to deal simultaneously with the acyclicity of $G[V\setminus S]$ and of $G[S]$. Say that a graph is an  \xclit
if there exists a linear order $\prec$ on the vertices of $G$ such that, for each vertex $v$, its successors are either 
all before $v$ ({\em i.e.} on its left side) or all after $v$ (on its right side) in the order $\prec$. 
Vertex $v$ is called a {\em left vertex} in the former case, and a 
{\em right vertex} in the latter case. The order is called an {\em LR-order}. The LR-order is {\em non-trivial} if it admits at
least one right vertex and at least one left vertex. An LR-order 
may be modified so as to move all the right vertices towards left (without changing their relative order) and all the left vertices towards right
(again, without changing their relative order)  after all the right vertices. The result is still an LR-order, that we call a {\em standard LR-order}.
The subgraph induced by each type of vertices is acyclic, since all the arcs outgoing from a right (resp. left) vertex  
are oriented towards right (resp. left).

The following claim is now easy:

\begin{fait}
$G$ has a non-empty acyclic FVS $S$ iff $G$ admits a non-trivial LR-order whose set of right vertices is $S$.
\label{claim:FVSxcl}
\end{fait}

\begin{proof} If $G$ has an acyclic FVS $S$, then $S\neq V$ and a topological order of $G[S]$ followed by a
reversed topological order of $G[V\setminus S]$ yields a non-trivial standard LR-order.
Conversely, let $\prec$ be a non-trivial LR-order and $S$ be the set of right vertices. Then $S\neq V$ and $G[S]$
is acyclic, since all the arcs in $S$ are oriented from left to right. The graph $G[V\setminus S]$
is also acyclic, since $V\setminus S$ is the set of left vertices and all arcs are oriented from right to left. \end{proof}
\bigskip

As a consequence, it is equivalent to look for an acyclic FVS in $G$, and to show that $G$ is an  \xclsans.
In order to be as close as possible to the problems we defined, and which have been defined in previous works,
we state the results in terms of acyclic FVS. However, in the proofs we merely look for an LR-order of $G$.

\begin{fait}
 Let $S$ be a set of variables from $U_3$. Then $S$ is the set of true variables in a NAE truth assignment of $\mathcal{F}$
 iff $S$ is an acyclic FVS of $G^\triangleleft(\mathcal{F})$.
 \label{claim:acyclic}
\end{fait}

\begin{proof} Recall that $G^\triangleleft(\mathcal{F})$ has vertex set $U_3$ of cardinality $N$, and has only 
strongly 3-covered cycles, by Claim~\ref{claim:tiny}.
An arbitrary clause of $\mathcal{F}$ is denoted  $(u_i^r\vee u_j^r\vee u_k^r)$ meaning
that literal $u_a^r$ is the occurrence of variable $u_a$ from $U_3$ (necessarily in positive form) in the $r$-th clause of $\mathcal{F}$. 
By definition,
in each clause, the order $\triangleleft$ defining the strongly 3-covered form is given by the order $u_i^r,u_j^r,u_k^r$
of the variables in the clause.

$\Rightarrow$: Consider a NAE truth assignment for the variables in $U_3$, whose true variables form the set $S$. 
For each arc $u_au_c$ of $G^\triangleleft(\mathcal{F})$ define the relation $u_a\blacktriangleleft u_c$ (intuitively: $u_a$ before $u_c$)
if $u_a$ is true, and 
$u_c\blacktriangleleft u_a$ if $u_a$ is false.   Now, build the directed graph $\Gamma$ with vertex set 
$U_3^\Gamma=\{u_1^\Gamma, \ldots, u_N^\Gamma\}$ and arc  set $\{u_e^\Gamma u_f^\Gamma| u_e\blacktriangleleft u_f\}$ 
(do not make use of transitivity). This graph collects the precedence relations defined by $\blacktriangleleft$,
and which require that true variables
be placed before their successors in $G^\triangleleft(\mathcal{F})$ and that false variables be placed after their successors in $G^\triangleleft(\mathcal{F})$.
More precisely:

\vspace*{-0.5cm}

\begin{multline}
 u_e\ \blacktriangleleft u_f\,  
\hbox{iff}\,\, \hbox{either literal}\, u_e\, \hbox{is true and}\, u_eu_f\, \hbox{is an arc of}\, G^\triangleleft(\mathcal{F})\, \\
\hbox{or literal}\, u_f\, 
\hbox{is false and}\, u_fu_e\, \hbox{is an arc of}\, G^\triangleleft(\mathcal{F}).
\label{eq:blackt}
\end{multline}

We show that $\Gamma$ is acyclic, so that any topological order of its vertex results into the sought LR-order $\prec$,
in which the true variables ({\em i.e.} those in $S$) are right vertices.

By contradiction, assume $\Gamma$ is not acyclic.
\medskip

(B1) {\em Let $q_1^\Gamma q_2^\Gamma \ldots q_d^\Gamma$ be a cycle in $\Gamma$. Then $q_1q_2\ldots q_d$ or $q_dq_{d-1}\ldots q_1$ 
is a cycle in  $G^\triangleleft(\mathcal{F})$.}
\medskip

Note first that $d\geq 3$, since otherwise the two arcs of a cycle of length two in $\Gamma$ would be defined by
two arcs in $G^\triangleleft(\mathcal{F})$ between the two same vertices, and $G^\triangleleft(\mathcal{F})$ has at most one arc between two given vertices
(Claim~\ref{claim:tiny}).
We make the convention that $q_{d+1}=q_1$, $q_0=q_{d}$ and similarly for the vertices in $\Gamma$. The reasoning is again 
by contradiction. Assume first that $q_1q_2$ is an arc of  $G^\triangleleft(\mathcal{F})$. Since by contradiction $q_1 q_2 \ldots q_d$ 
is not a cycle in $G^\triangleleft(\mathcal{F})$, let $h$ be the smallest 
index   such that $q_{h+1}q_h$ is an arc of  $G^\triangleleft(\mathcal{F})$. Then there
must also exist an index $l\geq h+1$ such that $q_{l}q_{l-1}$ and $q_lq_{l+1}$ are arcs of  $G^\triangleleft(\mathcal{F})$. Now, in $\Gamma$ the orientations 
of the two arcs with endpoint $q^\Gamma_l$ are defined by the truth value of $q_l$, and these arcs are both ingoing to $q_l^\Gamma$ 
(if $q_l$ is false) or both outgoing from $q_l^\Gamma$  (if $q_l$ is true).
In both cases,   we deduce that $q_1^\Gamma q_2^\Gamma \ldots q_d^\Gamma$ is not a cycle in $\Gamma$, contradicting the hypothesis. 
The reasoning is similar  in the case where $q_2q_1$ is an arc of  $G^\triangleleft(\mathcal{F})$.
\medskip

(B2) {\em No cycle $q_1q_2\ldots q_d$ in  $G^\triangleleft(\mathcal{F})$ induces a cycle in $\Gamma$.}
\medskip

Every cycle in $G^\triangleleft(\mathcal{F})$ is a strongly 3-covered cycle, since $\mathcal{F}$ is in strongly 3-covered form (Claim~\ref{claim:tiny}). 
Among the three literals of the same clause present 
in $q_1q_2\ldots q_d$, 
at least one (say $q_a$) is true and another one (say $q_b$) is false, according to the truth assignment. Then the arcs of $\Gamma$
with endpoints $q_a^\Gamma, q_{a+1}^\Gamma$ respectively  $q_b^\Gamma, q_{b+1}^\Gamma$ have opposite orientations, and the 
cycle of $G^\triangleleft(\mathcal{F})$ does not define a cycle in $\Gamma$.

\medskip

(B3) {\em $\Gamma$ has no cycle and any topological order of $\Gamma$ is an LR-order of $G^\triangleleft(\mathcal{F})$.}
\medskip

By (B1) a cycle in $\Gamma$ could only be defined by a cycle in $G^\triangleleft(\mathcal{F})$. However, by (B2) cycles in $G^\triangleleft(\mathcal{F})$
cannot define cycles in $\Gamma$. So $\Gamma$ is acyclic. Any topological order $\prec$ of $\Gamma$ extends the
relation $\blacktriangleleft$, and therefore satisfies: 

\begin{multline}
 u_a\prec u_b\,  \hbox{and there is an arc with endpoints}\,  u_a, u_b\,  \hbox{in}\,  G^\triangleleft(\mathcal{F})\, 
\hbox{iff}\\ \hbox{either literal}\, u_a\, \hbox{is true and}\, u_au_b\, \hbox{is an arc of}\, G^\triangleleft(\mathcal{F})\,
\hbox{or literal}\, u_b\, 
\hbox{is false and}\, u_bu_a\, \hbox{is an arc of}\, G^\triangleleft(\mathcal{F}).
\label{eq:prec}
\end{multline}

Now, let $u_c$ be a vertex of $G^\triangleleft(\mathcal{F})$. Then either literal $u_c$ is true and for each of its successors $u_d$ we have that
(by (\ref{eq:prec})) $u_c\prec u_d$ so that $u_c$ is a right vertex in $G^\triangleleft(\mathcal{F})$; or literal $u_c$ is false and for each 
of its successors $u_d$ we have that (by (\ref{eq:prec})) $u_d\prec u_c$, so that $u_c$ is a left vertex of $G^\triangleleft(\mathcal{F})$.
\bigskip

$\Leftarrow$: Assume now that $G^\triangleleft(\mathcal{F})$ admits an acyclic FVS $S$. Then $S\neq \emptyset$ since $G^\triangleleft(\mathcal{F})$ has cycles. 
By Claim~\ref{claim:FVSxcl},  $G^\triangleleft(\mathcal{F})$
admits an LR-order $\prec$ in which the right vertices are defined by $S$. Assign the 
value true to every variable that defines a right vertex in $G^\triangleleft(\mathcal{F})$ and the value false to every variable
that defines a left vertex in $G^\triangleleft(\mathcal{F})$. 

Every clause $(u_i^r\vee u_j^r\vee u_k^r)$ of  $\mathcal{F}$ defines
a cycle in $G^\triangleleft(\mathcal{F})$ with three vertices $u_i, u_j$ and $u_k$. In the LR-order $\prec$ of $G^\triangleleft(\mathcal{F})$, at
least one of these three vertices is a left vertex (the largest according to $\prec$) and at least one of 
them is a right vertex (the lowest according to $\prec$). Thus by assigning the value true to
all the vertices in $S$ ({\em i.e.} the right vertices of $G^\triangleleft(\mathcal{F})$), we define a NAE truth assignment for  
$\mathcal{F}$ whose true variables are exactly those in $S$.\end{proof}
\bigskip

Recall that $\tnae(\mathcal{F})$
denotes the minimum number of true variables in a NAE truth assignment of $\mathcal{F}$, if such an 
assignment exists.

\begin{fait}
Each NAE truth assignment for $\mathcal{F}$ (if such an assignment exists) with minimum number of true variables satisfies:
\begin{enumerate}
 \item[i)] for each $g\in U_2\setminus\{z\}$ exactly one of $\alpha_g$ and $\beta_g$ is true.
 \item[ii)] for each $g\in U_2\setminus\{z\}$, exactly one of the variables $a_g, b_g$ and $c_g$ is true. 
\end{enumerate}

\noindent Therefore $\frac{2(N-1)}{5}\leq \tnae(\mathcal{F})\leq \frac{2(N-1)}{5}+1$, with the minimum (resp. maximum) reached iff $z$ is false (resp. is true).

\label{claim:2N}
\end{fait}

\begin{proof} Affirmation {\em i)} is due to the consistency constraints and the NAE requirements. For Affirmation {\em ii)}, 
clause $(a_g\vee  b_g\vee c_g)$ and NAE requirements imply that at least one and at most two of the variables 
are true. Moreover, if for instance $a_g$ and $b_g$ are true in a NAE truth assignment,  modifying $b_g$ to false 
yields another NAE truth assignment with smaller number of true variables, a contradiction. 

By affirmations {\em i)} and {\em ii)}, a NAE truth assignment with a minimum number of true variables defines
as true  at least $\frac{2(N-1)}{5}$ variables, namely one variable among each pair $\alpha_g,\beta_g$
and one variable among each triple $a_g, b_g,c_g$. The unique remaining variable in $U_3$, the variable set of $\mathcal{F}$, 
is $z$. The claim follows. \end{proof}

\begin{fait}
  $\frac{2(N-1)}{5}\leq \fa(G^\triangleleft(\mathcal{F}))=\tnae(\mathcal{F})\leq \frac{2(N-1)}{5}$.
  \label{claim:tnae=fa}
\end{fait}

\begin{proof} By Claim~\ref{claim:acyclic},  a minimum acyclic FVS $S$ corresponds to the minimum set of true variables in a NAE truth 
assignment of $\mathcal{F}$, and viceversa. Therefore 
$\fa(G^\triangleleft(\mathcal{F}))=\tnae(\mathcal{F})$.
By Claim~\ref{claim:2N},  $\tnae(\mathcal{F})$ has the required bounds. \end{proof}

\subsection{NP-completeness results}\label{subsect:NP}

The results in the previous section allow us to easily deduce the hardness of {\sc Acyclic FVS} for 3c-digraphs:
\bthm
{\sc Acyclic FVS} is NP-complete even for 3c-digraphs.
\label{thm:xclreconn}
\ethm

\begin{proof} The problem is in NP, since testing the acyclicity of a graph is done in polynomial time.
Then the theorem is proved by reduction from {\sc 3-SAT}. Given a set $\mathcal{C}$ of
3-literal clauses for 3-SAT, its M-NAE version $\mathcal{F}$ satisfies Claim~\ref{claim:SATred},
{\em i.e.} 3-SAT with input $\mathcal{C}$ has a solution iff M-NAE 3-SAT with input $\mathcal{F}$ has a
solution. The last affirmation holds, by Claim~\ref{claim:acyclic}, iff $G^\triangleleft(\mathcal{F})$ has
an acyclic FVS. Moreover, by Claim~\ref{claim:tiny}, $G^\triangleleft(\mathcal{F})$ is a 3c-digraph,
and the theorem is proved. \end{proof}
\bigskip

The following claim will allow us to prove the hardness results for the {\sc 2-Choice} restrictions 
we defined.

\begin{fait}
The following affirmations are equivalent:
\begin{itemize}
\item[i)] $\tnae(\mathcal{F})=\frac{2(N-1)}{5}$
\item[ii)] $\ts(\mathcal{F})=\frac{2(N-1)}{5}$ 
\item[iii)] $\f(G^\triangleleft(\mathcal{F}))=\frac{2(N-1)}{5}$
\item[iv)]  $\fa(G^\triangleleft(\mathcal{F}))=\frac{2(N-1)}{5}$
\end{itemize}
\label{claim:4equiv}
 \end{fait}

 \begin{proof} By Claim~\ref{claim:ts=f}, $\f(G^\triangleleft(\mathcal{F}))=\ts(\mathcal{F})$ and thus Affirmations {\em ii)} and {\em iii)} are
 equivalent. Similarly, by Claim~\ref{claim:tnae=fa}, Affirmations {\em i)} and {\em iv)} are equivalent. It remains to show
 that {\em i)} and {\em ii)} are equivalent.

 $i) \Rightarrow ii)$: By Claim~\ref{claim:mintruth}, the hypothesis that $\tnae(\mathcal{F})=\frac{2(N-1)}{5}$ and since each NAE truth assignment is
 also a standard truth assignment, we have that $\frac{2(N-1)}{5}\leq \ts(F) \leq \tnae(\mathcal{F})=\frac{2(N-1)}{5}$
 and thus $\ts(\mathcal{F})=\frac{2(N-1)}{5}$. 
 
 $ii) \Rightarrow i)$: Since $\ts(\mathcal{F})=\frac{2(N-1)}{5}$, there exists a standard truth assignment $ta()$
 of $\mathcal{F}$ in which exactly one of $\alpha_g, \beta_g$ and exactly one of $a_g, b_g, c_g$, for each $g\in U_2\setminus\{z\}$, is true (by Claim~\ref{claim:mintruth} {\em i)} and {\em ii)}).
 Moreover $z$ is false  in this standard truth assignment, since the number of aforementioned true variables is  already $\frac{2(N-1)}{5}$.
 We now show that we can modify $ta()$ in order to obtain a NAE truth assignment containing the same number of true variables.
 To this end, notice first that if $\mathcal{F}$ contains clauses with three positive literals, these clauses cannot be among clauses
 $F_{1g}, F_{2g}, F_{3g}, F_{4g}$ and $F'_r$ since all of them contain at least one negative literal. 
 Then define the truth assignment $tb()$ for $\mathcal{F}$ as a variant of $ta()$ where, for each $r$ such that all literals are true
in the clause $F_r$, $\alpha_{w_r}$ is set to false and $\beta_{w_r}$ is set to true. Then $tb()$
is a NAE truth assignment with the same number of true literals as $ta()$, and thus by Claim~\ref{claim:2N} $tb(\mathcal{F})$ has minimum number of
true literals, thus gives the value of $\tnae(\mathcal{F})$. \end{proof}
\bigskip

We are now ready to prove the result concerning the  {\sc 2-Choice} variants we defined in Section \ref{sect:notations}.
Let us give the full name {\sc Strongly 3-Covered 2-Choice-Min1-M 3-SAT} to the problem {\sc 2-Choice-Min1-M 3-SAT}
where the input is restricted to clauses $\mathcal{C}$ admitting a 3-strongly dominated form, and
similarly for {\sc 2-Choice-Min1-M-NAE 3-SAT}. Then the representative graph of the set of clauses in the input is a 3c-digraph.

\bthm
The following problems are NP-complete:  
\begin{itemize}
 \item {\sc Strongly 3-Covered 2-Choice-Min1-M 3-SAT}
 \item {\sc Strongly 3-Covered  2-Choice-Min1-M-NAE 3-SAT}
 \item {\sc 2-Choice MFVS} on 3c-digraphs
 \item {\sc 2-Choice Acyclic MFVS} on 3c-digraphs.
\end{itemize}

\label{thm:2choice}
\ethm

\begin{proof} All problems belong to NP, since they are restrictions of problems belonging to NP.

The reduction is done from NAE 3-SAT, which is NP-complete. Let the input of NAE 3-SAT  be a set $\mathcal{C}$ of 
3-literal clauses over the variable set $X=\{x_1, x_2, \ldots, x_n\}$. As we did above for the instances of 3-SAT
(which are the same as those for 3-SAT), we build the M-NAE version $\mathcal{F}$ of $\mathcal{C}$. Then we have:
\bigskip

(C1) {\em {\sc NAE 3-SAT} with input $\mathcal{C}$ has a solution iff\, {\sc M-NAE 3-SAT}  with input $\mathcal{F}$ has a solution where $z$ is false.}
\bigskip

Indeed, if NAE 3-SAT with input $\mathcal{C}$ has a solution then each clause $(l_i\vee l_j\vee l_k)$ has at least one true literal (say $l_p$)
and at least one false literal (say $l_q$). Following the reductions in Table \ref{table:red}, recall 
the interpretation of the literals $h_u$: $l_u$ is true in the initial instance iff $h_u\neq z$ in the NAE 4-SAT instance. 
Then, if $z$ is false, $h_p$ is true and $h_q$ is false, implying that the clause $(h_i\vee h_j\vee h_k\vee z)$
of NAE 4-SAT is true. The other reductions successively transform this clause into equivalent sets clauses, implying that
there is a solution of M-NAE 3-SAT  where $z$ is false. Conversely, if $z$ is false in the solution of M-NAE 3-SAT for $\mathcal{F}$, 
then by the equivalence of the clauses at each successive reduction (taken backwards in Table \ref{table:red}) we deduce
that in each clause $(h_i\vee h_j\vee h_k\vee z)$ of NAE 4-SAT some $h_p$ must be true whereas some $h_q$ must be false.
The relationship between $l_u$ and $h_u$ then implies that $l_p$ is true and $l_q$ is false, therefore   NAE 3-SAT with 
input $\mathcal{C}$ also  has a solution.
\bigskip

(C2) {\em {\sc M-NAE 3-SAT} with input $\mathcal{F}$ has a solution where $z$ is false iff $\tnae(\mathcal{F})=\frac{2(N-1)}{5}$.}
\bigskip

This is an easy consequence of Claim~\ref{claim:2N}.
\bigskip

(C3)  {\em $\tnae(\mathcal{F})=\frac{2(N-1)}{5}$ iff {\sc Strongly 3-Covered 2-Choice-Min1-M-NAE 3-SAT} with input $\mathcal{F}$ and
$D=\frac{2(N-1)}{5}$ has a positive answer.}
\bigskip

We first notice that  $\mathcal{F}$ satisfies the hypothesis of {\sc Strongly 3-Covered 2-Choice Min1-M-NAE 3-SAT}.
Indeed $\mathcal{F}$ has only strongly 3-covered cycles 
by Claim~\ref{claim:tiny}, and $\frac{2(N-1)}{5}\leq \tnae(\mathcal{F})\leq \frac{2(N-1)}{5}+1$ by Claim~\ref{claim:2N}.
The affirmation is then immediate, by the definition of the problem and the bounds we have on $\tnae(\mathcal{F})$.
\bigskip

From the equivalences proved in (C1), (C2) and (C3), we deduce that {\sc Strongly 3-Covered 2-Choice-Min1-M-NAE 3-SAT}
is NP-complete. To deduce that the three other problems are NP-complete, the same approach is used. 
By  Claim~\ref{claim:4equiv}, (C2) holds even when in its second affirmation $\tnae()$ is replaced 
with $\ts()$, or $\f()$ or $\fa()$. And (C3) holds when 1) $\tnae()$ is replaced with $\ts()$, or $\f()$ or $\fa()$,
and 2) {\sc Strongly 3-Covered 2-Choice Min1-M-NAE 3-SAT} is respectively replaced with {\sc Strongly 3-Covered 2-Choice Min1-M 3-SAT}, or
{\sc 2-Choice MFVS} on 3c-digraphs, or {\sc 2-Choice Acyclic MFVS} on 3c-digraphs. \end{proof}

\section{Classes for which {\em Acyclic FVS} is polynomially solvable}\label{sect:classes}

According to Claim~\ref{claim:FVSxcl}, solving {\sc Acyclic FVS} on a particular class
of directed graphs is equivalent to testing in polynomial time, for the graphs in the class, whether an LR-order
exists. 

In this section, we present two classes of \xcls for which the LR-order always exists and may
be found in polynomial time.  {\sc Acyclic FVS} is therefore solvable in polynomial time, 
and a solution is computable in polynomial time. 

\subsection{ \consecs}

A directed simple graph $G=(V,E)$ is called a  \consec  if its 
adjacency matrix has the Consecutive Ones Property for Rows (C1PR). That means the columns of the matrix can be reordered
such that in each row of the matrix the ones are consecutive. An order of the columns that correctly positions the ones in each 
row is called a {\em good order}.
Then we easily have:

\begin{fait}
  \consecs are exactly the \xcls admitting an LR-order positioning the successors of each vertex on consecutive places. 

\end{fait}

\begin{proof} Let $G$ be a  \consec  and $A$ be its adjacency matrix. Let $A^{C1R}$ be the matrix
obtained from $A$ by reordering the columns according to the good order, and then reordering the
rows according to the good order. Note that the row reordering does not change the C1R property
realized by the column reordering. Then $A^{C1R}$ still defines $G$ but with the vertices ordered according
to the good order. In this order, the successors of each vertex $v$ (represented by the ones on the row
corresponding to $v$) are consecutive. Moreover, since $G$ is loopless, the element on the main diagonal
of  $A^{C1R}$ is 0 in each row, therefore the successors of $v$ are either all before or all after the diagonal.
Then $v$ is a left vertex if the ones on its row are before the diagonal, and a right vertex otherwise.

Conversely, it is straightforward that each \xcl admitting an LR-order positioning the successors of each vertex 
on consecutive places is a \consecsans. \end{proof}
\bigskip

Note that each  order  positioning the successors of each vertex on consecutive places necessarily is an LR-order. To
find such an order, we take advantage of the numerous studies on the C1P property (see \cite{Dom} for a survey). 
They show that a good order may be obtained in linear  time for such a graph, using for instance the algorithm in \cite{Booth}. 

\consecs may be defined in terms of intersections of intervals, as follows. Following \cite{das}, consider a family $V$ of
ordered pairs of intervals $(S_v,T_v)$, $v=1, 2, \ldots,n$, on the real line. The {\em intersection digraph} of the
family $V$  is the digraph with vertex set $\{1, 2, \ldots, n\}$ and whose arcset is defined by all the ordered pairs 
$vw$ such that $S_v\cap T_w\neq\emptyset$. 
An {\em interval digraph} is any intersection digraph of a family of ordered pairs of intervals.  When the intervals
$T_v$ are singletons, the interval digraph is called an {\em interval-point digraph}. Then we have:

\begin{fait}[\cite{das}]
 A digraph $D$ is an interval-point digraph iff its adjacency matrix has the Consecutive Ones Property for rows.
\end{fait}

As \consecs are simple digraphs, we deduce that \consecs are exactly the loopless interval-point digraphs, {\em
i.e.} those for which the ordered pairs of intervals $(S_v, \{t_v\})$ satisfy $t_v\not\in S_v$ (so that $vv$
is never an arc).

\subsection{Reducible flow graphs}

A directed graph $G=(V,E)$ is a {\em flow graph} if there exists a vertex $s$ such that any vertex in $V\setminus\{s\}$ is
reachable by a path from $s$. The flow graph is then denoted $(G,s)$. Flow graphs
are used to model program flows and are therefore important in practice \cite{Shamir}. Among the well-known subclasses of flow graphs,
rooted DAGs and series-parallel directed graphs are very well studied examples. We consider here the class of reducible flow graphs,
introduced in \cite{Hecht2} and studied for instance in \cite{Hecht,Tarjan,Aho}. This class is one of the 
(not so many) classes of directed graphs for which MFVS may be solved in polynomial time \cite{Shamir}.
However, the FVS with minimum size returned by the algorithm in \cite{Shamir} may not induce an
acyclic graph, as shown by the graph in Fig. \ref{fig:contrex} (which is inspired from an example proposed in \cite{Shamir}).

\begin{figure}
\centering
 \includegraphics[angle=0, width=5cm]{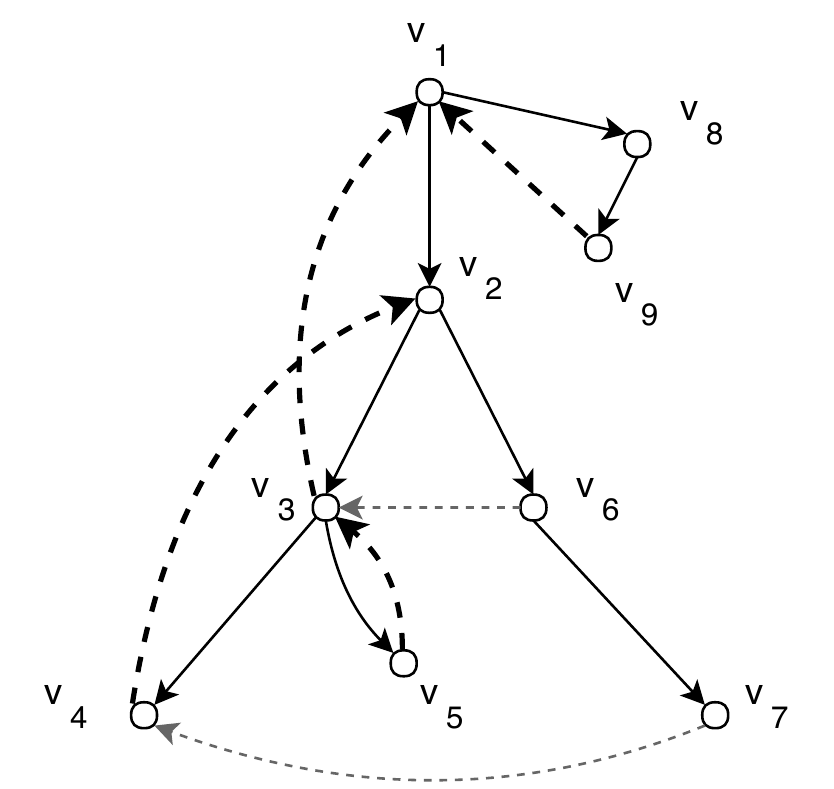}
\label{fig:contrex}
 \caption{Example of a graph for which the minimum FVS computed by the algorithm in \cite{Shamir}
 induces a cycle. Plain black, bold dashed and dotted arrows respectively represent tree, cycle and cross arcs. The minimum FVS
 returned by the algorithm in \cite{Shamir} is formed by the vertices $v_1, v_2$ and $v_3$, which induce a cycle.}
 \end{figure}

Most of the definitions we give below follow those in \cite{Tarjan}.
Let $(G,s)$, with $G=(V,E)$, be a flow graph. We assume $G$ is a simple directed graph, {\em i.e.} it does not have loops or 
multiple arcs. A {\em preorder numbering} of $G$ is an assignment $po()$ of
numbers to the vertices of $G$ according to a depth-first search traversal of $G$: $po(s)=1$ and 
$po(v)>po(w)$ iff $v$ is visited after $w$. Notice that we do not require consecutive $po()$ values.
The directed tree $T=(V,E(T))$ defined by the traversal is called a {\em depth-first spanning tree} (or {\em DFST}) of
$G$. The existence of a (directed) path in $T$ from $v$ to $w$ is denoted $v\pathh w$. 
Then the arc set $E(T)$ is partitioned into:

\begin{itemize}
 \item {\em tree arcs} $vw$, satisfying $vw\in E(T)$.
  \item {\em forward arcs} $vw$, satisfying $vw\not\in E(T)$ and $v\pathh w$.
 \item {\em cycle arcs} $vw$, satisfying $w\pathh v$.
 \item {\em cross arcs} $vw$, such that none of $v\pathh w$ and $w\pathh v$ holds, and $po(w)<po(v)$.
\end{itemize}

Several equivalent definitions of reducible flow graphs are proposed \cite{Hecht}. We chose to say that a flow 
graph $(G,s)$ is {\em reducible} if its acyclic subgraph induced by the union of  tree, forward and cross arcs 
is invariable when $T$ changes, {\em i.e.} has the same set of arcs for any DFST $T$ rooted at $s$. 
Reducible flow graphs have a lot of nice properties, that we need for our proof and that we recall below.
We use the notation $(G,s,T)$ for a flow graph $(G,s)$ provided with a DFST $T$. The associated preorder numbering
is denoted $po()$.

\subsubsection{Properties of reducible flow graphs}\label{sect:propredflow}

For this first property, say that a vertex $w$ {\it dominates} a vertex $v$ if $w\neq v$ and each path 
from $s$ to $v$ in $G$ contains $w$.

\begin{fait}[\cite{Hecht}]
Let $(G,s,T)$ be a flow graph. Then $(G,s)$ is reducible iff $w$ dominates $v$ for any cycle arc $vw$. 
\label{claim:domin}
\end{fait}

The second property confirms an intuition issued from the definition:

\begin{fait}[\cite{Hecht}] 
The cycle arcs of a reducible flow graph $(G,s)$ are invariable for any DFST $T$ rooted at $s$.
\label{claim:cyclearcs}
\end{fait}

Let $w$ be a vertex of the flow graph $(G,s,T)$. Then define 
$C(w)=\{v\,|\, vw\,\, \hbox{is a cycle arc}\}$ and $P(w)=\{v\,|\, \exists z\in C(w)\,\, \hbox{such that there is
a path from}\, v\, \hbox{to}\, z\, \hbox{which avoids}\, w\}$. If $C(w)$ is non-empty and $w\neq s$, then
$w$ is called a {\em head}. This definition is correct, because of Claim~\ref{claim:cyclearcs}. Vertex $s$
is not called a head even if $C(w)$ is not empty since it has a special role, as root of $T$.
Furthermore,  say that a vertex $w$ satisfies the {\em $T$-path property} if for all $v\in P(w)$, $w\pathh v$.

\begin{fait}[\cite{Tarjan}]
The flow graph $(G,s,T)$ is reducible iff each head $w$ satisfies the $T$-path property.
\label{claim:Tarjan}
\end{fait}

Notice that $s$ trivially satisfies the $T$-path property, by the definition of $T$. 
Let $(G,s)$ be a flow graph and $v, w$ with $v\neq s$ be two of its vertices. {\em Collapsing
$v$ into $w$} in $G$ is the operation that ``shrinks'' $v$ and $w$ into a single vertex called $w$,
removing loops and redundant arcs. More formally, arcs are added from $w$ to each vertex in $N^+(v)\setminus N^{+}[w]$ and
from each vertex in $N^-(v)\setminus N^-[w]$ to $w$. Then $v$ and its incident arcs are removed. {\em Collapsing a set 
$S\subseteq V\setminus \{s\}$ into $w$} in $G$ is the operation that successively collapses each $v\in S$
into $w$. Given two vertices $v, w$, a {\em reduction} collapses $v$ into $w$ if
$N^-(v)=\{w\}$ and $v\neq s$.

\begin{fait}[\cite{Hecht}]
 A flow graph $(G,s)$ is reducible iff it can be transformed into the graph $(\{s\},\emptyset)$ by a series of reductions.
\label{claim:reds}
 \end{fait}

In order to show that reducible flow graphs are \xclssans, we need to explain the algorithm in \cite{Tarjan} for 
recognizing reducible flow graphs and constructing the sequence of reductions (according to Claim~\ref{claim:reds}).  

\subsubsection{Tarjan's reducibility algorithm}\label{sect:Tarjansalg}

The basic idea of the algorithm is issued from  Claim~\ref{claim:Tarjan}: $(G,s)$ is reducible iff each head $w$ has the 
$T$-path property.  The algorithm proposed in \cite{Tarjan} uses this idea in an optimized way, by
successively collapsing sets of vertices, so that subsequent tests are easier to perform. This is possible
since the collapsed graph inherits useful properties and constructions of the initial graph:

\begin{fait}[\cite{Tarjan}]
 Let $(G,s,T)$ be a flow graph and $w_1$ be its head with largest $po()$ value. Let $G'$ be the graph resulting from $G$
 by collapsing $P(w_1)$ into $w_1$. Then:
 \begin{itemize}
 \item[i)] $(G',s)$ is a flow graph.
  \item[ii)] For each arc $v'u'$ in $G'$, either $v'u'$ is an arc of $G$, or $v'=w_1$ and there is a vertex $u$ in $G$ such that
  $uu'$ is an arc of $G$ and  $w_1\pathh u$ in $T$ ($v'u'$ in $G'$ then {\em corresponds} to the arcs $v'u'$ of $G$ in the first case, 
  resp. to the arc $uu'$ of $G$ in the second case). 
  \item[iii)] The subgraph $T'$ of $G'$ whose arcs correspond to arcs in $T$  is a DFST of $G'$, with the preorder numbering given by the
  $po()$ values (restricted to the vertices in $G'$).
  \item[iv)] Cycle, forward, cross arcs of $G'$ respectively correspond to cycle, forward, cross arcs of $G$.
  \item[v)] If $C'(w)$ and $P'(w)$ are defined in $(G',s,T')$ similarly to $(G,s,T)$, then the heads of $G'$ are the
  same as the heads of $G$ except $w_1$. Moreover, for each head $w$ of $G'$, $w$ has the $T'$-path property in $G'$ iff $w$
  has the $T$-path property in $G$.
 \end{itemize}
 \label{claim:resteflow}
\end{fait}

In the Reducibility algorithm (Algorithm \ref{algo:reducibility}),  heads $w$ are ordered in decreasing order of their
$po()$ value (step 1). For each $w$ in this order, the algorithm tests the $T$-path property  in the {\em current graph} 
using the sets  $P^*(w)$ of vertices (step 5 to 10), collapses $P^*(w)$ into $w$ and continues with the next head.  
In order to recall the head a vertex $v$ is collapsed into, parameter $hn(v)$ is defined to be $po(w)$ for all  
$v\in P^*(w)$ (step 8). In order to homogenize the presentation, we define $hn(v)=1$ for all $v\neq s$ that 
belong to no $P^*(w)$ and we denote $P^*(s)=\{v\in V\, |\, hn(v)=1\}$. Then $P^*(s)$ is collapsed into $s$
(step 13), and thus when the graph is reducible, in the end the resulting graph is $(\{s\},\emptyset)$.
Note that the reductions are not detailed in this algorithm, but collapsing $P^*(w)$ into $w$ at each step  
means that a sequence of reductions is able to successively collapse  each $v$ from $P^*(w)$ into $w$, as explained later.  

\begin{algorithm}[t]
\caption{Reducibility algorithm \cite{Tarjan}}
\begin{algorithmic}[1]
\REQUIRE A flow graph $(G,s)$.
\ENSURE  Answer ``No'' if $(G,s)$ is not reducible; otherwise, answer ``Yes'' and for each vertex $v\neq s$ the value $hn(v)$.

\STATE Let $w_1, w_2, \ldots, w_k$ be the heads of $G$, ordered by decreasing value of $po()$.
\STATE {\bf for} $v\in V\setminus\{s\}$ {\bf do} $hn(v)=1$ {\bf endfor}
\STATE $G^{w_0}\leftarrow G$ \hfill{\sl //$w_0$ does not exist, this is a simple notation}
\FOR{$i=1, 2\ldots, n$}
\STATE Let $P^*(w_i)$ be defined as $P(w_i)$ in the current graph $G^{w_{i-1}}$
\FOR{$v\in P^*(w_i)$}
\STATE {\bf if} not $w_i\pathh v$ {\bf then} return ``No'' {\bf endif}
\STATE $hn(v)\leftarrow po(w_i)$
\ENDFOR
\STATE Collapse $P^*(w_i)$ into $w_i$, and call $G^{w_i}$ the resulting graph
\ENDFOR
\STATE Let $P^*(s)=\{v\in V\setminus\{s\}\, |\, hn(v)=1\}$
\STATE Collapse $P^*(s)$ into $s$.
\STATE Return ``Yes''
\end{algorithmic}
\label{algo:reducibility}
\end{algorithm}

\begin{rmk}
 Sets $C^*(w)$ and $P^*(w)$ are subsets of $V$, although they are defined in a modified graph. Moreover, sets $P^*(w)$
 are disjoint.
 
 \label{rmk:p*}
\end{rmk}

\bex
Consider the flow graph in Fig. \ref{fig:filrouge}. The four heads $w_1, w_2, w_3$ and $w_4$ are indicated on the
figure. Algorithm Reducibility computes $P^*(w_1)=\{e,f\}$, $P^*(w_2)=\{c,d\}$, $P^*(w_3)=\{w_1,w_2\}$, $P^*(w_4)=\{a,b,w_3\}$
and $P^*(s)=\{w_4\}$. The values $hn()$ are computed on this basis, and are indicated on the figure (second integer in the triple associated with
each vertex).
\eex

The following relationship may be established between $P^*(w)$ and $P(w)$:

\begin{fait}
Let $(G,s,T)$ be a reducible graph. Then $v\in P^*(w_i)$ iff 
$v\in P(w_i)$ and $po(w_i)=\max\{po(w_j)\, |\, v\in P(w_j)\}$. 
Moreover, the vertices $w_j$ such that $v\in P(w_j)$ appear  
on the path from $s$ to $v$, in decreasing order of their value $j$.
\label{claim:p*max}
\end{fait}

\begin{proof}  We first show that if $v\in P^*(w_i)$ then $v\in P(w_i)$.
Since $v\in P^*(w_i)$, there exists a path $P$ in $G^{w_{i-1}}$  avoiding $w_i$ 
which joins $v$ to a vertex $z\in C^*(w_i)$. By Claim~\ref{claim:resteflow}{\em ii)} applied to
$w_{i-1}, w_{i-2}, \ldots, w_1$, each arc of $P$ may be successively replaced by paths in 
$G^{w_{i-2}}, G^{w_{i-3}}, \ldots, G^{w_1}, G$ so as to obtain a path in $G$ from $v$ to $z$.
Moreover, in  Claim~\ref{claim:resteflow}{\em ii)}, the path in $T$  from $w_1$ to $u$ cannot 
contain any $w_i$, $i>1$, since then the head $w_1$ would not have the largest value $po()$. 
This implies that the resulting path in $G$ also avoids $w_i$, and thus $v\in P(w_i)$. 

We are now ready to prove the claim.

$\Leftarrow$:  By Claim~\ref{claim:resteflow}{\em v}) $w_i$ remains a 
head in $G^{w_1}, \ldots,G^{w_{i-1}}$. Moreover, by the property we just proved, $v$ cannot belong to $P^*(w_1), \ldots, P^*(w_{i-1})$,
therefore $v$ is a vertex of $G^{w_{i-1}}$. Now, $v\in P(w_i)$ implies the existence of a path $P$ in $G$ avoiding $w_i$
and joining $v$ to some $z$ in $C(w_i)$. The resulting paths in $G^{w_1}, \ldots, G^{w_{i-1}}$ also avoid $w_i$
(since no vertex is collapsed into $w_i$) and join $v$ to $z$ (or the vertex $z$ collapses into).
But then $v\in P^*(w_i)$ since $C^*(w_i)$ contains $z$ (or the vertex $z$ collapses into), as the arc from
$z$ (or the vertex it collapses into) to $w_i$ remains a cycle arc (Claim~\ref{claim:resteflow}{\em iv})). 

$\Rightarrow:$ By the property at the beginning of the proof, we have that $v\in P(w_i)$.
If by contradiction $po(w_i)\neq\max\{po(w_j)\, |\, v\in P(w_j)\}$, then let $w_k$ be the head reaching the maximum value. 
According to the $\Leftarrow$ part of the claim, $v\in P^*(w_k)$. This is in contradiction
with  Remark \ref{rmk:p*} indicating that the sets $P^*(w_j)$ are disjoint. 

To finish the proof, by Claim~\ref{claim:Tarjan}, the property $w_j\pathh v$ holds for all $j$ such that
$v\in P(w_j)$. Since there is a unique path from $s$ to $v$ in $T$, all these heads $w_j$ are on this
path. As $po(w_j)>po(w_l)$ iff $j<l$ by definition, it follows that the order of the heads $w_j$ on the
path from $s$ to $v$ is the decreasing order of the index $j$. \end{proof}

\begin{figure}[t]
\centering
 \includegraphics[angle=0, width=7cm]{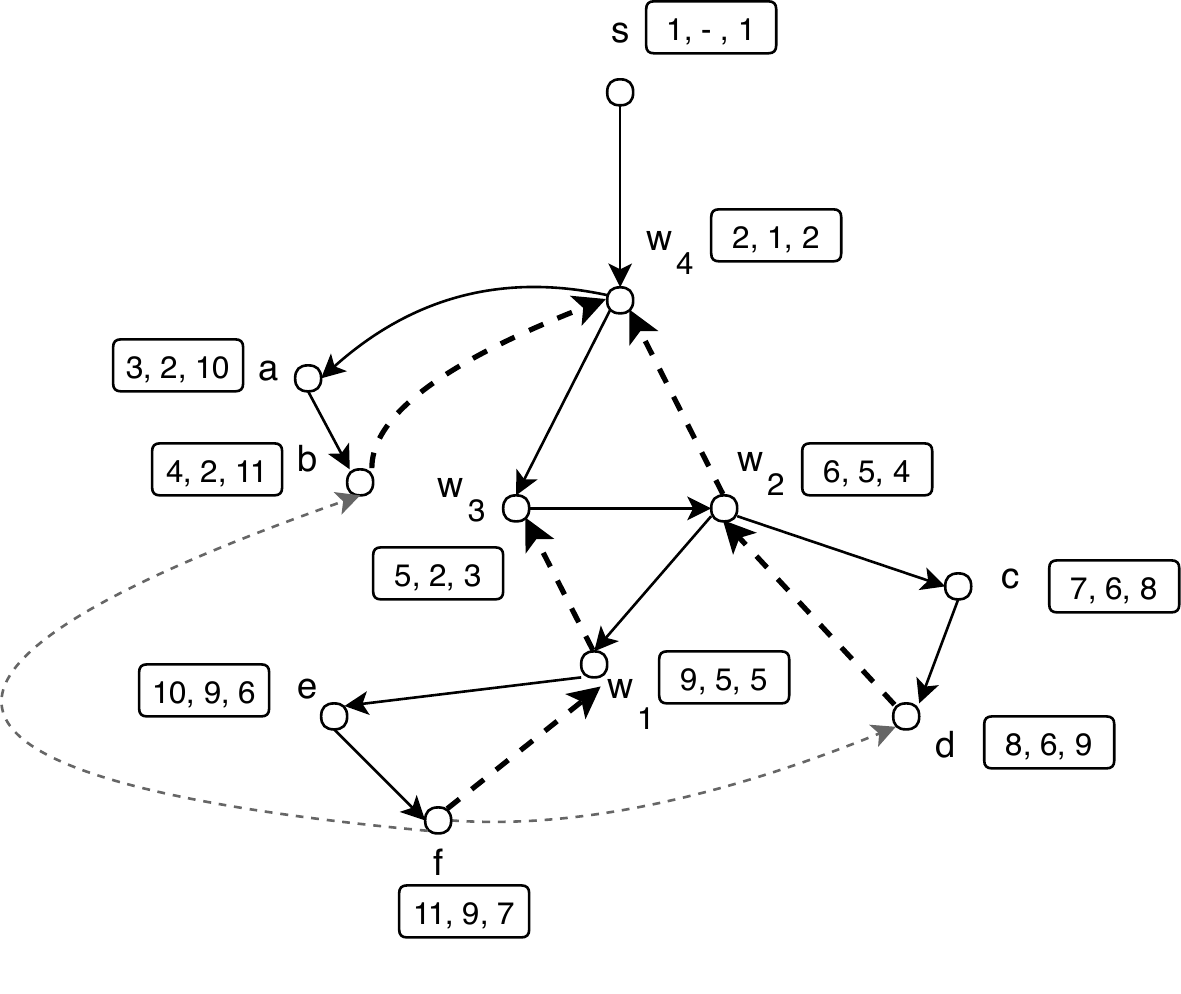}

 \caption{\label{fig:filrouge}An example of flow graph. The three integer values next to each vertex respectively
represent its $po(), hn()$ and $sn()$ numbers (the definitions are given in the text when we need them). As before,
plain black, bold dashed and dotted arrows respectively represent tree, cycle and cross arcs.}
 \end{figure}
 
%
\bigskip

Once Algorithm Reducibility is applied, the reduction order is computed as follows. First, perform a preorder traversal of $T$ 
(already built) by considering the children of each vertex in decreasing order of their $po()$ numbers. This new
preorder traversal assigns new numbers $sn(v)$ to the vertices $v$ such that \cite{Tarjan}:

\begin{equation}
 \label{eq:0}
sn(v)<sn(w)\, \hbox{for an arc}\,  vw\, \hbox{iff}\, vw\, \hbox{is a tree, forward or cross arc.}
\end{equation}

Then the reduction order $\alpha$ is established using the couples of values
$(hn(v),sn(v))$, for each $v\in V$, as follows: vertex $v$ appears before vertex $t$ in the reduction order
iff either $hn(v)>hn(t)$, or $hn(v)=hn(t)$ and $sn(v)<sn(t)$. Equivalently, all vertices collapsed into $w_1$
are before all the vertices collapsed into $w_2$ and so on. Vertices collapsed into $s$ (those with $hn(v)=1$)
are at the end of the reduction sequence. Moreover, vertices collapsed into the same vertex $w_i$
are ordered according to their increasing value $sn()$. Then, each vertex represents the reduction with its father 
in the tree $T'$ resulting from the previous reductions. This father is  exactly $w$ such that $v\in P^*(w)$.

\bex
On the example in Fig. \ref{fig:filrouge}, the resulting order $\alpha$ gives the sequence of vertices:
$e,f,c,d$, $w_2,w_1, w_3,$ $a, b, w_4,s$.
\label{ex:alpha}
\eex

\begin{rmk}
In the following, we call Reducibility+ the algorithm obtained from Reducibility by adding the necessary instructions
allowing it to output the heads $w_i$ ordered by decreasing value of $po()$, as well as the sets $P^*(s)$ and, for all $i$,  $P^*(w_i)$.
\end{rmk}

\subsubsection{Our algorithm for finding an LR-order}

A cycle arc $vw$ of $(G,s)$ is called {\em solved} if $v$ is collapsed
into $w$ (i.e. if $hn(v)=po(w)$) by Algorithm Reducibility, and {\em unsolved} otherwise. In the latter case, 
$w$ is also a head, and there exists a path in $T$ from $w$ to $hn(v)$ (since 
the cycle arcs from $v$ to $w$ and to $hn(v)$ means they are both on the path from $s$ to $v$; by definition,
$hn(v)$ is the one with higher $po()$ number, so the lower one). 

The LR-ordering algorithm (Algorithm \ref{algo:LR}) attempts to affect the vertices of $G$ either to 
the set $R$ (right vertices) or to the set $L$ (left vertices) following the idea that for each $w$ 
which is either a head or $s$, the vertices in $P^*(w)$, on the one hand, and $w$, on the other hand, should
be in different parts of the partition $(R,L)$. Instead of working with vertices, the algorithm starts by working
with {\em blocks}, that are sets of vertices. The blocks are the singletons containing one vertex $w$ each 
(where $w$ is a head or $s$) and the non-empty sets $P^*(w)\setminus W$ (so that each vertex of $G$ belongs to a unique block).
The algorithm thus builds (steps 3 to 5) the {\em undirected} graph $B$ whose vertices are the blocks and whose
edges join vertices that must be in different parts $(R,L)$. This graph turns out to be a tree (see Claim~\ref{claim:tree}),
and thus admits a partition in two sets of blocks $R'$ and $L'$ with no internal edges, computed in step 6.
Once the partition is found, $R$ (resp. $L$) collects all the vertices in some block of $R$ (resp. of $L$).
The graphs $G[R]$ and $G[L]$ turn out to be acyclic (Claim~\ref{claim:acyclic2} below), and thus each of
them is an acyclic FVS for $G$. In order to find an LR-order of $G$, it is sufficient to order
the vertices in each part according to the topological order, to reverse the order found for $G[L]$,
and to concatenate the two resulting sequences of vertices (steps 9 to 11).
%

\bex
On the example in Fig. \ref{fig:filrouge}, the boxes are $\{w_1\}$,  $\{w_2\}$, $\{w_3\}$, $\{w_4\}$,
$\{e,f\}$(=$P^*(w_1)\setminus W$),  $\{c,d\}$(=$P^*(w_2)\setminus W$) and $\{a,b\}$(=$P^*(w_4)\setminus W$).
The set $P^*(w_3)\setminus W$ is empty, and is not a vertex of $B$. The edges of $B$ are therefore joining
$\{w_1\}$ to $\{e,f\}$ (due to $w_1$ and $P^*(w_1)$), $\{w_2\}$ to $\{c,d\}$ (due to $w_2$), $\{w_3\}$ to $\{w_1\}$
and $\{w_{2}\}$ (due to $w_3$), $\{w_4\}$ with $\{w_3\}$ and $\{a,b\}$ (due to $w_4$), and $\{s\}$ to $\{w_4\}$
(due to $s$). The partition $(R',L')$ of $B$ is thus:

$R'=\{\{w_1\}, \{w_2\},\{w_4\}\}$

$L'=\{\{e,f\},\{c,d\},\{w_3\}, \{a,b\}, \{s\}\}$

\noindent or viceversa. Then $R=\{w_1, w_2,w_4\}$ and $G[R]$  has only two arcs, from $w_2$ to the other vertices, 
meaning that  one can choose for instance the topological order $R^*$ given by $w_2, w_1,w_4$. Similarly,
$L=\{e,f,c,d,w_3, a,b, s\}$, with arcs $ab,ef,cd, fb, fd$. The topological order on $G[R]$ may be
chosen to be, for instance, $s, e$, $f, a, b, c,$ $d, w_3$ which yields, after a complete reversal,
$L^*:w_3, d, c, b, a, f, e, s$.  The LR-order $U$ is then $w_2,w_1,w_4$, $w_3,d,$ $c,b,a,f,e, s$, with the three first vertices being
right vertices and the remaining ones being left vertices.
\eex

Note that, in order to avoid confusions, we always use the term {\em box} to designate the vertices of $B$,
and we reserve the term {\em vertex} for the vertices of $G$. 
The edges of the undirected graph $B$ then 
join each box $\{w\}$ (corresponding to some $w\in W$) with the boxes $\{w'\}$ whose unique vertex belongs to $P^*(w)$, as well as 
to the box $P^*(w)\setminus W$ containing the other elements in $P^*(w)$. 

\begin{algorithm}[t]
\caption{LR-ordering algorithm}
\begin{algorithmic}[1]
\REQUIRE A reducible flow graph $(G,s,T)$.\\
\ENSURE A sequence $U$ of the vertices in $G$, defining an LR-order of $V$.\\

\STATE Apply Algorithm Reducibility+ on $(G,s,T)$.
\STATE $W\leftarrow \{w_1, \ldots, w_k,s\}$
\STATE $V'\leftarrow \cup_{w\in W} \{\{w\}, P^*(w)\setminus W\}$ \hfill // if $P^*(w)\setminus W=\emptyset$, do not use it
\STATE $E'\leftarrow\{\{w\}X\,|\, w\in W, X=P^*(w)\setminus W\, \hbox{or}\,X=\{w'\}\,\, \hbox{with}\,\, w'\in W\cap P^*(w)\}$ \hfill //edges, not arcs
\STATE $B\leftarrow (V',E')$ 
\STATE $(R',L')\leftarrow$ partition of $V'$ such that each of $B[R'], B[L']$ is edgeless
\STATE $R\leftarrow \{x\in V\,|\, x\, \hbox{belongs to a box in}\, R'\}$
\STATE $L\leftarrow \{w\in V\,|\, x\, \hbox{belongs to a box in}\, L'\}$
\STATE $R^*\leftarrow $  a topological order of the vertices in $G[R]$.
\STATE $L^*\leftarrow $  a reversed topological order of the vertices in $G[L]$.
\STATE $U\leftarrow$ concatenate $R^*$ and $L^*$ in this order.
\STATE Return $U$.
\end{algorithmic}
\label{algo:LR}
\end{algorithm}

\bthm
Algorithm LR-ordering computes in polynomial time an LR-order of a reducible flow graph $(G,s)$.
\label{thm:reducible}
\ethm 

\begin{proof}  It is clear that the algorithm runs in polynomial time, since all the operations it 
performs are polynomial. The proof is organized in three claims, showing that the behavior of the algorithm
is the one we expected in our previous explanations.

\begin{fait}
 The (undirected) graph $B$ defined in steps 3-5 of Algorithm LR-ordering is a tree.
\label{claim:tree}
\end{fait}

\begin{proof} Each box $X$ in $B$ but $\{s\}$ is collapsed exactly once in Algorithm Reducibility, since by definition
each box $X$ is a subset of some $P^*(w_{i(X)})$. According to the definition of $E'$, the edges of $B$ are exactly the pairs 
$\{w_{i(X)}\}X$. Then each box $X$ but $\{s\}$ has exactly one father $w_{i(X)}$, and $B$ is a tree. \end{proof}.

\begin{fait}
 The graphs $G[R]$ and $G[L]$ are acyclic.  
\label{claim:acyclic2}
 \end{fait}

\begin{proof} The proof is similar for $G[R]$ and $G[L]$. We therefore present it only for $G_R$.
Let $xy$ be an arc of $G[R]$.

\bigskip

(D1)  {\em If $xy$ is a cycle arc then $sn(x)>sn(y)$ and $hn(x)>hn(y)$. Otherwise,  $sn(x)<sn(y)$ and $hn(x)\geq hn(y)$.}

\bigskip

{\bf Case 1.} Consider first the case where $y$ collapses into a head $w\neq s$. Then $y\in P^*(w)$ and
thus by Claim~\ref{claim:p*max} we have $y\in P(w)$. Let $wb_1\ldots b_i(=y)\ldots b_l(=v)$ be 
the cycle of $G$ obtained by concatenating the path  from $w$ to $y$ in $T$ (see Claim~\ref{claim:Tarjan}),
 a path avoiding $w$ that joins $y$ with some $v\in C(w)$ and the arc $vw$. We show that $w$ dominates $x$ in $G$. 
Indeed, if a path $P_1$ from $s$ to $x$ avoiding $w$ existed,  then the path $P=P_1b_i(=y)b_{i+1}\ldots b_l(=v)$ would be a path 
from $s$ to $v$ avoiding $w$, a contradiction to Claim~\ref{claim:domin}.
Then $w$ dominates $x$, and thus $w$ is on the path $P'$ in $T$ joining $s$ to $x$. Then $x$ collapses into one of the vertices in $P'$
(by Claim~\ref{claim:p*max}). If it collapses into one of the vertices $z$ between $w$ (non-included) and $x$, then $hn(x)>hn(y)$ since
$po(z)>po(w)$. Otherwise, $x$ belongs to $P^*(w)$ since there is a path from $x$ to $v(=b_l)$, namely $xb_i(=y)b_{i+1}\ldots b_l(=v)$,
and $w$ satisfies $po(w)=\max\{po(w)\, |\, x\in P(w)\}$ (see Claim~\ref{claim:p*max}). 
In this case, $hn(x)=hn(y)$. Thus in all cases $hn(x)\geq hn(y)$. However, the equality cannot occur when $xy$ is a cycle arc. Indeed, if $xy$ is a cycle arc, then $y$
is on the path in $T$ from $s$ to $x$ (by the definition of a cycle arc) and moreover $xy$ is unsolved. Then $x$ cannot collapse 
into $y$ and 
must collapse into a vertex $w'$ situated on the path from $y$ to $x$ in $T$. But then $hn(x)>hn(y)$ since $po(w')>po(w)$.
Property (\ref{eq:0}) of $sn()$ finishes the proof in this case.

{\bf Case 2.} In this case, $y$ collapses into $s$, thus necessarily $hn(y)=1$ and thus $hn(x)\geq hn(y)$. Again, for an unsolved cycle arc 
$xy$ the equality cannot occur since then $x$ must collapse into some $w$ which necessarily has larger $po()$ than $s$, and thus 
has strictly larger $hn()$.
According to property (\ref{eq:0}), $sn(x)<sn(y)$ for an arc $xy$ iff $xy$ is a tree, forward or cross arc and we are done.

\bigskip

(D2) {\em $G[R]$ cannot contain a cycle $A=a_1a_2\ldots a_t$.}

\bigskip

By (D1), if such a cycle exists, arcs $a_ja_{j+1}$ ($j=1,\ldots, t-1$) imply that $hn(a_1)\geq hn(a_2)\geq \ldots \geq hn(a_t)$.
If there is at least one strict inequality, we deduce $hn(a_1)>hn(a_t)$ and thus the arc $a_ta_1$ does not satisfy (D1), a contradiction.
Therefore, none of the arcs  $a_ja_{j+1}$ ($j=1,\ldots, t$), where by convention $a_{t+1}=a_1$, is a cycle arc, as cycle arcs
satisfy $hn(a_j)> hn(a_{j+1})$. But then again by (D1) we have that $sn(a_1)< sn(a_2)< \ldots < sn(a_t)$, and the arc
 $a_ta_1$ does not satisfy (D1). \end{proof}
\medskip

\begin{fait}
 The sequence $U$ gives a standard LR-order of $G$.
\end{fait}

\begin{proof} This is obvious now, since $G[R]$ and $G[L]$ are acyclic, and $U$ is built using their topological orders.
\end{proof}

Theorem~\ref{thm:reducible} is now proved. \end{proof}

\begin{rmk}
On the example in Fig. \ref{fig:contrex}, Algorithm LR-ordering finds $R=\{v_1, v_3, v_4, v_6,v_7\}$ and $L=\{v_2, v_5, v_8, v_9\}$,
each of which induce an acyclic graph. Note that none of the acyclic FVS $R$ and $L$ solves {\sc Acyclic MFVS}, since a
minimum size acyclic FVS with three vertices exists (take for instance $\{v_2, v_5, v_8\}$). 
\end{rmk}

\section{Conclusion}\label{sect:conclusion}

In this paper, we investigated feedback vertex set problems, both from the viewpoint of their hardness and by proposing 
easier particular cases. We have shown close relationships between these problems, in their standard or acyclic variant,
and the 3-SAT problems, in their standard or NAE variant. As a result, we showed the NP-completeness of {\sc Acyclic FVS}
even on the class of 3c-digraphs. We have also shown close relationships between the minimal variants of the aforementioned  problems,
and deduced that even the choice between two proposed values that are one unit far from each other is NP-hard.
And this holds even on the class of 3c-digraphs (which are the input graphs in the case of FVS problems, and the
representative graphs of the clauses provided with an order of the literals, in the case of 3-SAT problems).

Many questions remain open however. Is the class of 3c-digraphs a hard case for other NP-complete problems? If so,
what structural properties justify it? Can we extend the NP-hardness of the {\sc 2-Choice} variants we have studied
to smaller classes of graphs? Are the \xcls an important or useful class of graphs, for which - for instance - 
other NP-hard problems than {\sc Acyclic FVS} have polynomial solutions? To start with, is it possible to solve
MFVS or {\sc Acyclic MFVS} in polynomial time on \xclssans? 
Or are there many other classes of graphs ({\em e.g.} the
cyclically-reducible,  the quasi-reducible graphs or the completely contractible graphs that we cited in the
Introduction) that are subclasses of \xclssans? 


\bibliographystyle{plain}
\bibliography{XClass}
\end{document}